%% file: main.tex
\pdfoutput=1

\newif\ifdraft\drafttrue 
\makeatletter \@input{texdirectives.tex} \makeatother

\PassOptionsToPackage{names,dvipsnames}{xcolor}
\documentclass[acmsmall,screen]{acmart}

\setcopyright{rightsretained}
\acmJournal{PACMPL}
\acmYear{2022}
\acmVolume{6}
\acmNumber{ICFP}
\acmArticle{124}
\acmMonth{8}
\acmPrice{}
\acmDOI{10.1145/3547655}
\copyrightyear{2022}
\acmSubmissionID{icfp22main-p107-p}

\usepackage[T1]{fontenc} %
\usepackage[utf8x]{inputenc}
\input{definitions}
\definecolor{shadecolor}{gray}{0.93}

\usepackage{booktabs} 
\usepackage{enumitem}
\usepackage{pifont}

\newcommand{\subs}[2]{[#1 / #2]}

\definecolor{defgreen}{rgb}{0,0.6,0}

\crefname{section}{\textsection\!}{\textsection\!}
\crefname{proposition}{Prop.}{Props.}

\begin{document}

\renewcommand{\footnotemark}{\mbox{}}

\title{A Reasonably Gradual  Type Theory}

\author{Kenji Maillard}
\affiliation{%
  \institution{Gallinette Project-Team, Inria}
  \city{Nantes}
  \country{France}
}
\orcid{0000-0001-5554-3203}
\email{kenji.maillard@inria.fr}
\author{Meven Lennon-Bertrand}
\affiliation{%
  \institution{Gallinette Project-Team, Inria}
  \city{Nantes}
  \country{France}
}
\orcid{0000-0002-7079-8826}
\email{meven.lennon-bertrand@inria.fr}
\author{Nicolas Tabareau}
\affiliation{%
  \institution{Gallinette Project-Team, Inria}
  \city{Nantes}
  \country{France}
}
\orcid{0000-0003-3366-2273}
\email{nicolas.tabareau@inria.fr}
\author{\'Eric Tanter}
\affiliation{%
  \institution{PLEIAD Lab, Computer Science Department (DCC), University of Chile}
  \city{Santiago}
  \country{Chile}
}
\orcid{0000-0002-7359-890X}
\email{etanter@dcc.uchile.cl}

\titlenote{This work is partially funded by CONICYT FONDECYT Regular
Project 1190058 and Inria {\'E}quipe Associ{\'e}e GECO.}



\begin{abstract}
Gradualizing the Calculus of Inductive Constructions (CIC) involves dealing with subtle tensions between normalization, graduality, and conservativity with respect to CIC. Recently, GCIC has been proposed as a parametrized gradual type theory that admits three variants, each sacrificing one of these properties. 
For devising a gradual proof assistant based on CIC, normalization and conservativity with respect to CIC are key, but the tension with graduality needs to be addressed. 
Additionally, several challenges remain:
(1)~The presence of two wildcard terms at any type---the error and unknown terms---enables trivial proofs of any theorem, jeopardizing the use of a gradual type theory in a proof assistant;
(2)~Supporting general indexed inductive families, most prominently equality, is an open problem;
(3)~Theoretical accounts of gradual typing and graduality so far do not support handling type mismatches detected during reduction;
(4)~Precision and graduality are external notions not amenable to reasoning within a gradual type theory.
All these issues manifest primally in \CCIC, the cast calculus used to define \GCIC. 
In this work, we present an extension of \CCIC called \CCICPrec. \CCICPrec is a reasonably gradual type theory that addresses the issues above, featuring internal precision and general exception handling.
\shnew{\CCICPrec features an impure (gradual) sort
of types inhabited by errors and unknown terms, and a pure (non-gradual) sort of
strict propositions for consistent reasoning about gradual terms.}
By adopting a novel interpretation of the unknown term that carefully accounts for universe levels, 
\CCICPrec satisfies graduality for a large and well-defined class of terms, 
in addition to being normalizing and a conservative extension of CIC.
Internal precision supports reasoning about graduality within \CCICPrec itself, for instance to characterize gradual exception-handling terms, and supports gradual subset types.
We develop the metatheory of \CCICPrec using a model formalized in \Coq,
and provide a prototype implementation of \CCICPrec in \Agda.
\end{abstract}

\maketitle

\section{Introduction}
\label{sec:introduction}

Extending gradual typing~\cite{siekTaha:sfp2006,siekAl:snapl2015} to dependent types is a challenging endeavor due to the intricacies of type checking and conversion in the presence of imprecision at both the type and term levels.
Early efforts looked at gradualizing specific aspects of a
dependent type system (\eg subset types and
refinements~\cite{lehmannTanter:popl2017,tanterTabareau:dls2015}, or the
fragment without inductive types~\cite{eremondiAl:icfp2019}). Recently, \citet{lennonAl:toplas2022} studied gradual typing in the context of the Calculus of Inductive Constructions (\CIC), the theory at the core of many proof assistants such as Coq~\cite{Coq:manual}.
%

\shnew{
\paragraph{Gradual \CIC}
\citet{lennonAl:toplas2022} develop a gradualization of \CIC, called \GCIC. 
For instance, as in simply-typed gradual typing, one can use the unknown type $\?$ to defer some checks to runtime: $(\lambda x:\?.\ x+1)\;v$ is well-typed for any $v$, and may reduce to a runtime error if $v$ is not a natural number.}
\GCIC is a source language, whose semantics is given by elaboration to a dependently-typed cast calculus, called \CCIC. 
\CCIC is an extension of Martin-Löf type theory (\MLTT)~\cite{itt} with (non-indexed)
inductive types, and with exceptions as introduced by
\citet{pedrotTabareau:esop2018}.
For a given type $A$, there are two
exceptional terms, namely $\err_A$ representing runtime type errors, and $\?_A$
representing the {\em unknown term}, which can optimistically stand for any term of type $A$. \shnew{In particular, the unknown type is $\?_{[]}$, where $[]$ denotes the universe (omitting levels for brevity here).}
Additionally, \CCIC features a cast operator $\cast{A}{B}{t}$, which supports treating a
term $t$ of type $A$ as a term of type $B$, without requiring any relation between $A$ and $B$.
\shnew{The above example in \GCIC elaborates to the \CCIC term
$(\lambda x:\?_{[]}. \cast{\?_{[]}}{\nat}{x}+1)\;\cast{V}{\?_{[]}}v$, where $V$ is the type of $v$.
If $v$ is $10$, this term reduces to 11; if $v$ is $\btrue$, the term reduces to $\err_\nat$.
The dependently-typed setting involves a number of peculiarities and complexities, which come from the fact that there are unknown terms at all types, and that gradual computation can happen at the type level as well.
}

\paragraph{Variants of Gradual \CIC}

Crucially, \citet{lennonAl:toplas2022} uncover an inherent tension in the gradualization of \CIC, dubbed the Fire Triangle of Graduality, which states that three fundamentally desirable properties cannot be fully satisfied simultaneously: (1) strong normalization, a property of particular relevance in the context of proof assistants, (2) conservativity with respect to \CIC, namely the ability to faithfully embed the static theory in the gradual theory,
and (3) graduality, which guarantees that typing and evaluation are monotone with respect to precision.\footnote{In the gradual typing literature, graduality is first known as the gradual guarantees~\cite{siekAl:snapl2015}; the dynamic aspect thereof was later reformulated by \citet{newAhmed:icfp2018} under a more semantic form, which turns out to be stronger than the dynamic gradual guarantee in the setting of dependent types~\cite{lennonAl:toplas2022}.}


{\em Precision} is an essential notion in gradual typing~\cite{siekAl:snapl2015}, which captures the expected behavior of casts:
when a type $A$ is more precise than $B$, 
written $A \precision{}{} B$, then
casting from $A$ to $B$ does not fail, and doing the roundtrip back to
$A$ is the identity; 
the formal formulation of this property, coined
{\em graduality} by \citet{newAhmed:icfp2018}, is that when 
$A \precision{}{} B$,
the cast operations induce an embedding-projection pair
between $A$ and $B$. Additionally, $\?$ is the least precise type, 
and therefore casting from $A$ to the unknown type $\?$
and back is always the identity.
The maximality of the unknown type is a key element in the tension captured by the Fire
Triangle of Graduality. Indeed, if $\? -> \? \precision{}{} \?$, then by graduality it is possible 
to embed the untyped lambda calculus, and in particular the diverging term 
$\Omega := (\lambda~ x : \?.~x~x)~(\lambda~ x : \?.~x~x)$.

To study different resolutions of the Fire Triangle in a unified framework, 
\citet{lennonAl:toplas2022} develop \GCIC as a {\em parametrized} gradualization of \CIC.
\GCIC admits three variants, each sacrificing one property: 
\GCICG satisfies both conservativity and graduality at the expense of admitting divergence, 
\GCICN dynamically avoids non-termination but this carefulness inevitably leads to some terms violating graduality, and finally, \GCICs restricts the typing relation of \CIC to exclude those non-gradual terms and hence satisfies graduality and termination but does not admit all \CIC terms.
\CCIC is itself parametrized, yielding \CCICG, \CCICN, and \CCICs as dependent cast calculi underlying each of the three \GCIC variants.

\paragraph{Termination and Universe Levels}
In \GCIC, the unknown type is the unknown term at the universe type, $[]$. But due to predicativity in \CIC there is in fact an infinite hierarchy of universes $[]_i$. This means that in \GCIC there is one unknown type per level of the stratification; each $\?_{[]_i}$ is the least precise type among all types at level $i$ and below.
The two \GCIC variants that ensure termination avoid divergence by shifting universe levels either statically (\GCICs) or dynamically (\GCICN). 
\GCICs restricts the typing rule of the function type compared to vanilla \CIC
by incrementing the universe level of the function type with respect to that of its components. 
%
Its main downside is that it is not a conservative extension of \CIC: due to this modified typing rule, some valid \CIC terms are statically rejected. The prototypical example is that of {\em recursive large elimination}, such as the type of n-ary functions over natural numbers (in \Coq):

\begin{coq}
  Fixpoint nArrow (n : nat) : Type_0 := match n with 0 => nat | S m => nat -> narrow m.
\end{coq}

\noindent The term \coqe{nArrow n} is a type (\ie a term of type $[]_0$), and we have for example
\coqe{nArrow 0 equiv nat} and \coqe{nArrow 2 equiv nat -> nat -> nat}.
The reason this definition is ill-typed in \CICs is that the universe level at
which to define the resulting type is unbounded. Another more practical example
is that of a dependently-typed \coqe{printf} function, whose actual arity
depends on the input string. \shnew{Still, \GCICs captures a large and useful
  fragment of \CIC, which includes most examples of functional programs found in
  predicative System~F and also uses of dependent types where large elimination has a statically-known bound.}

In the context of a gradual proof assistant based on \CIC, the normalizing and conservative variant \GCICN is therefore the most appealing, as it ensures decidability of typing, (weak) canonicity, and supports all existing developments and libraries by virtue of being a conservative extension of \CIC. \GCICN avoids non-termination by introducing a universe shift during reduction, which unfortunately means that some terms break graduality. 
For instance, while \coqe{nArrow} is well-typed in \GCICN, 
the type \coqe{forall (n:nat), nArrow n} does not satisfy the embedding-projection property with respect to any unknown type $\?_{[]_i}$, because the appropriate universe level is not known {\em a priori}.
However, apart from the fact that \GCICN does not satisfy graduality globally, little is known about its gradual properties as its metatheory in this regard has not been developed. In particular, there is no clear characterization of a class of terms for which graduality holds.


\paragraph{A Refined Stratification of Precision} 
In this work, we observe that by refining the stratification of
precision we can develop a full account of graduality for an extension
of \CCICN, called \CCICPrec. The key idea is that $\?_{[]_i}$ should be the least precise type among all types at level $i$ and below, {\em except} for dependent function types at level $i$ (which are however still less precise than $\?_{[]_{i+1}}$).
We can precisely characterize problematic terms as those that are not {\em self-precise} (\ie~more precise than themselves). As we will see, for function types, self-precision means monotonicity with respect to precision. A recursive large elimination as in \coqe{nArrow} is not monotone because,
\shnew{
  $\mathtt{nArrow}\; ?_{\nat}$ computes to $\?_{[]_i}$ for some fixed level $i$, but there is
  no $i$ such that $\mathtt{nArrow}\; n \precision{}{} ?_{[]_i}$ uniformly for all $n$.
}
We prove that the dynamic gradual guarantee holds in \CCICPrec for any self-precise context, and that casts between types related by precision induce embedding-projection pairs between self-precise terms.
Therefore, this change in perspective in the interpretation of the unknown type
and the associated notion of precision yields a gradual theory that
conservatively extends \CIC, is normalizing, and satisfies graduality for a large and well-defined class of terms. \shnew{Specifically, we prove that all terms that would be well-typed with a level-shifting dependent product type (as used by \GCICs/\CCICs) can be embedded in \CCICPrec and proven to be self-precise, and hence satisfy graduality. Also, some terms that fall
outside of that fragment can be proven self-precise in \CCICPrec.}

\paragraph{Internalizing Precision, Reasonably}
While we could study graduality for \CCICPrec externally, we observe that we can exploit the expressiveness of the type-theoretic setting to internalize precision and its associated reasoning. In particular this makes it possible to state and prove, within the theory itself, results about (self-)precision and graduality for specific terms. 
\shnew{For such internal reasoning to be reliable, \CCICPrec adopts a two-layer structure, with an impure hierarchy of types for gradual terms, and a pure sort of propositions that can refer to gradual terms and errors, but whose inhabitants cannot use errors or unknown terms. This approach to isolate effects is inspired by prior approaches to soundly reason about effectful programs internally with dependent 
types~\cite{pedrotTabareau:popl2020,swamyAl:popl2016,stumpAl:par2010,kimmellAl:plpv2012,CasinghinoAl:popl2014} (discussed in \cref{sec:related-work}), most notably the Reasonably Exceptional Type Theory \RETT~\cite{pedrotAl:icfp2019}.} \RETT supports consistent reasoning about exceptional terms by featuring
 a layer of possibly exceptional terms, and a separate layer of pure terms in which raising an exception is prohibited. This way, the consistency of the logical layer is guaranteed, while allowing non-trivial interaction with the exceptional layer. Technically, the two layers are defined using two distinct universe hierarchies.

Additionally, internalizing precision requires the gradual type theory to satisfy extensionality principles in order to support the notion of precision as error approximation~\cite{newAhmed:icfp2018}. To this end, \CCICPrec builds upon the observational type theory \TTOBS~\cite{pujet:hal-03367052}.
Based on the seminal work on Observational Type Theory~\cite{altenkirchAl:plpv2007}, 
\TTOBS provides a setoidal equality in a specific universe $\prop$ of definitionally proof-irrelevant propositions. 
This universe of strict propositions, introduced by \citet{gilbert:hal-01859964} and supported in recent versions of \Coq and \Agda, 
makes it possible to define an extensional notion of equality, while trivializing the so-called higher
coherence hell by imposing that any two proofs of a given equality are {\em definitionally} equal.
The resulting theory is arguably much simpler and closer to the current practice of proof assistants than 
cubical type theory~\cite{cubicaltt,cubical-agda}, which is another approach to provide extensional principles with computational content.

A major insight of this work is to realize that we can actually merge the logical universe of \RETT 
used to reason about exceptional terms with the universe $\prop$ of proof-irrelevant propositions in order to 
define an internal notion of precision that is extensional and whose proofs cannot be trivialized with exceptional terms.

\paragraph{Applications of Internal Precision}
Being able to internally reason about the graduality of terms in a theory that is not globally gradual \shnew{is essential for a gradual proof assistant. Because precision semantically accounts for error approximation~\cite{newAhmed:icfp2018}, internal precision provides a useful reasoning principle to certify gradual programs. Just like internal equality enables reasoning using Leibniz equality (\ie~deducing that $P\; b$ holds given both $P\; a$ and $a = b$), internal precision makes it possible to deduce the correctness of a gradual program from the correctness of another: if we have $a \precision{}{} b$ and $P\;a$ for a correctness criterion $P$ that is self-precise and thus monotone, then $P\;b$ holds.} 
\shnew{
For instance, 
consider the following two functions related by precision:
$$
\mathtt{add1} := \lambda x : \nat. x + 1 \quad \precision{}{}\quad
\mathtt{add1?} := \lambda x : \?_{[]}. (\cast{\?_{[]}}{\nat}{x}) + 1
$$
The term $t := \mathtt{map}\,\nat\,\nat\,\mathtt{add1}\, l$ is fully static and hence does not fail, given a non-error list $l : \listT\,\nat$. Now, to show that the term $u := \mathtt{map}\,\?_{[]}\,\nat\,\mathtt{add1?}\, l'$ (where $l'$ is $\cast{\listT\,\nat}{\listT\,\?_{[]}}{l}$) also does not fail, one can either reason directly on the definition of $u$, or one can deduce the property ``for free'' from the fact that $t \precision{}{} u$, which follows from the monotony of $\mathtt{map}$.}\footnote{\shnew{The fact that $\mathtt{map} : \Pi A B : []. (A \to B) \to \listT\,A \to \listT\,B$ is self-precise and hence monotone with respect to all its arguments is proven by simple induction on lists. See the Agda development for details of this example.}}

Additionally, internal precision makes it possible to support gradual subset types, in which a type can be refined by a proposition expressed using precision. Moreover, in the literature, exception handling is never considered when proving graduality because this mechanism inherently allows terms that do not behave monotonically with respect to precision. Internal precision enables us to support exception handling in the impure layer of the type theory, and to consistently reason about the graduality (or not) of exception-handling terms.

\paragraph{Structure of the Article}
We propose \CCICPrec, a novel gradual type theory with internal precision and a
two-layer architecture that enables consistent reasoning about potentially
failing and imprecise gradual programs. \CCICPrec is a
strongly-normalizing extension of \CIC that satisfies graduality for a large and well-defined class of terms.
After a brief informal overview of the main elements of \CCICPrec and their applications (\cref{sec:action}), we formalize \CCICPrec as an extension of \CCIC  with a sort of propositions (\cref{sec:ccic}) and a precision relation for internal reasoning about graduality (\cref{sec:precision}). We present a model of \CCICPrec in \CIC, which
validates its metatheoretical properties (\cref{sec:model}). 
\Cref{sec:impl-deta-agda}~discusses extensions of \CCICPrec and \cref{sec:related-work} reviews related work. 
We provide a
\href{https://gitlab.inria.fr/kmaillar/grip-a-reasonably-gradual-type-theory}{\Coq
  formalization of the model and a proof-of-concept implementation in \Agda}
(\href{https://zenodo.org/record/6928465}{artifact after evaluation}).


\section{A Brief Overview of \CCICPrec}
\label{sec:action}

\CCIC has been introduced by \citet{lennonAl:toplas2022} as a
variant of \CIC with exceptional terms and a cast operator, 
designed to support the source gradual type theory \GCIC. 
Due to the use of conversion for typing in dependently-typed systems, 
\GCIC requires elaboration into \CCIC 
for both its static and dynamic semantics. This elaboration, which introduces casts as necessary to account for imprecision in \GCIC terms, is not the focus of this work; instead, we tackle issues at the level of the design and semantics of the type theory with casts, \CCIC. After a quick refresher on \CCIC, this section introduces the two-layer architecture of \CCICPrec for consistent reasoning about gradual programs, the notion of internal precision and its application to reason about graduality, including in the presence of exception handling, and gradual subset types.

\subsection{Background on \CCIC}
\label{sec:ccic-background}

Technically, \CCIC features an impure hierarchy of universes $[]_i$ (read ``Type'') where one
can freely use unknown terms, noted $\?_A$ for any type $A$, and
errors, noted $\err_A$.
The hierarchy $[]_i$ is explicitly cumulative, meaning that there is a
constructor $\cum : []_i -> []_{i+1}$ that permits to consider a type
at level $i$ as a type at level $i+1$.
\CCIC also features inductive types such as natural numbers (noted $\nat$),
booleans (noted $\bool$) and lists of elements of type $A$ (noted $\listT~A$).
The only difference with the corresponding inductive types in \CIC
is that there are two additional constructors for each inductive
type, one corresponding to errors $\err$ and the other to the unknown term $\?$ at that type.  
Additionally, \CCIC features casts, whose typing rule is
\begin{mathpar}
    \inferrule
    { {\Gamma} \vdash {A} \ccicty {[]_i} \\
      {\Gamma} \vdash {B} \ccicty {[]_i} \\
      {\Gamma} \vdash {t} \ccicty {A}}
    {{\Gamma} \vdash {\cast{A}{B}{t}} \ccicty {B}}
\end{mathpar}
A cast converts any term of type $A$ to a term of type $B$, with no
constraint between $A$ and $B$. This means that a cast propagates
deeper when types are compatible, \eg two function types:
$$
\ascdom{f}{A_1 -> B_1}{A_2 -> B_2}
\quad \redCCIC \quad 
\l y :
A_2. \ascdom{(f~\cast{A_2}{A_1}{y})}{B_1}{B_2}
$$
But when $A$ and $B$ are
not compatible, a cast reduces to an error in $B$, \eg between booleans
and natural numbers, we have 
$
\cast{\bool}{\nat}{\btrue} \redCCIC \err_{\nat}
$. \shnew{Following \citet{pedrotTabareau:esop2018}, both $\?$ and $\err$ behave like
\emph{call-by-name} exceptions.}
\shnew{In particular, this means that $(\lambda x:\nat. 0)\;\err_{\nat} \redCCIC 0$, not $\err_{\nat}$. Also, exceptions can only be caught on positive types such as inductives, not on negative types such as functions. Notably, $\err_{\P x : A . B } \conv \lambda x:A. \err_B$.
}

The main features of \CIC that are absent in \CCIC are an impredicative universe of propositions and a
general notion of indexed inductive types.

\subsection{A Universe for Logical Reasoning}

Directly inspired by the work on the {\em reasonably} exceptional type theory RETT~\cite{pedrotAl:icfp2019},
\CCICPrec features two distinct kind of sorts: the impure hierarchy of
types $[]_i$ of \CCIC, and a pure impredicative sort of definitionally proof-irrelevant propositions $\prop_{}$. While
propositions can be {\em about} gradual terms and errors, they cannot
be themselves inhabited by unknown terms or errors, thereby ensuring
consistent logical reasoning. 
\citet{lennonAl:toplas2022} show that no good notion
of equality can be defined in the impure hierarchy of types because
of an unsolvable tension between canonicity and the reduction of cast on
equality. In \CCICPrec, the absence of imprecision in $\prop$ means the cast operator 
does not need to be defined between propositions, and therefore the tension disappears.

To be able to reason about properties of inductive types in $\prop$,
their elimination principles needs to be extended for predicates in
$\prop$.
However, contrarily to predicates valued in the impure hierarchy of types,
there is no default behavior for errors and $\?$. Thus eliminators in
$\prop$ require additional arguments to deal with those two
exceptional cases, in a way reminiscent of try-catch for exception handling.
For instance, the eliminator for $\bool$ (if-then-else) is given by:
$$
\bcatchProp : \forall (P: \bool -> \prop), P~\btrue -> P~\bfalse ->
P~\err_\bool -> P~\?_\bool -> \forall (b : \bool), P~b
$$
In this logical layer, it becomes possible to reliably prove properties, because 
it is not possible to prove a false result in $\prop$ by means of the unknown (or error) term, contrarily to $[]$. For instance, we can prove
that casting from $\bool$ to  $\nat$ is always an error, stated as 
$
\forall (b : \bool), \cast{\bool}{\nat}{b} = \err_\nat
$. This result is proven by a direct use of reflexivity of equality because
the cast simply reduces to an error.

\subsection{Internal Precision}
\label{sec:internal-prec}

\CCICPrec features internal precision as an heterogeneous relation in the pure logical universe $\prop_{}$, defined
between gradual types and terms of gradual types, as expressed by the typing rules:
  \begin{mathpar}
    \inferrule
    {\Gamma \vdash A,B : []_i}
    {\Gamma{} \vdash A \precisionType{i} B : \prop_{}}
    \and
    \inferrule
    { \Gamma \vdash A,B : []_i \\
    \Gamma{} \vdash t : A \\
      \Gamma{} \vdash u : B}
    {\Gamma{} \vdash t \precision{A}{B} u : \prop_{}}
\end{mathpar}
Because the universe level at which gradual types are defined plays a central
role in the definition of precision, we explicitly annotate type precision 
with the level at which it occurs. 
Note that precision on proofs of propositions is undefined: there is no way 
to be imprecise in the logical layer.

\citet{garciaAl:popl2016} describe a systematic approach to design gradual languages, in which precision follows from the interpretation of gradual types as the set of static types that they denote. For instance, the type $\nat -> \?$ denotes all function types with $\nat$ as domain; this type is deemed more precise than the unknown type $\?$ because the latter denotes any type. Therefore, precision among types coincides with the set inclusion of their denotations. 
Of course, in the context of a stratified hierarchy of types, with full dependency, the situation is more challenging.

To better reflect the semantics of \CCICN with respect to universe levels during reduction, 
which avoids diverging terms such as $\Omega$ without affecting typing,
in \CCICPrec we adjust the denotation of the unknown type at
universe level $i$, $\?_{[]_i}$, so that it excludes dependent
function types at level $i$.  Consequently, at level $i$, all type
constructors except functions are more precise than $\?_{[]_i}$,
so the following propositions hold 
(mentioning only
lists as the prototypical example of inductive types):
%
%
\begin{mathpar}
    \inferrule
    {}
    { []_i \precisionType{i+1} \?_{[]_{i+1}}}
    \and
    \inferrule
    {}
    { \listT~A \precisionType{i} \? _{[]_{i}} \text{ whenever }  
    A \precisionType{i} \? _{[]_{i}} }
    \and
    \inferrule
    {}
    {\cum A \precisionType{i+1} \? _{[]_{i+1}}} \and
        \inferrule
    {}
    { \?_{[]_i} \precisionType{i} \?_{[]_{i}}} 
  \end{mathpar}
In particular, in order to be more precise than the unknown type, a dependent function type needs to be \emph{guarded} by an explicit use of cumulativity with $\cum : []_i -> []_{i+1}$.
This means that we can derive $\cum (\nat -> \nat)
\precisionType{1} \?_{[]_1}$ and $\cum (\?_{[]_0} -> \?_{[]_0})
\precisionType{1} \?_{[]_1}$, but $\nat -> \nat
\not \precisionType{0} \?_{[]_0}$ and $\?_{[]_0} -> \?_{[]_0}
\not \precisionType{0} \?_{[]_0}$.

Once the
definition of precision on the unknown type is fixed, the rest of the
definition is naturally obtained from
congruence/extensional rules.
 We do not detail here the definition of internal term precision (presented in \cref{sec:precision})
but, for instance, precision between two functions 
$f \precision{\forall a, B~a}{\forall a', B'~a'} g$ boils down to pointwise 
precision:
 $
 \forall a\, a',~ a \precision{A}{A'} a' \to f\,a \precision{B~a}{B'~a'} g\,a'
$.
The only remaining subtlety is the definition of term precision 
in the impure sort $[]_i$, as it should be connected to type precision, 
because terms of $[]_i$ are types. 
Precision on types, when seen as terms of the sort $[]_i$, is
the restriction of type precision to types that are
more precise than $\?_{[]_{i}}$, \ie
$
 A \precision{[]_i}{[]_i} B  \text{ corresponds to } A \precisionType{i} B \wedge B
 \precisionType{i} \?_{[]_{i}}
$.

Consequently, \CCICPrec has the global property that $\?_A$ is maximal for \emph{term precision}
of any type $A$, even when $A$ is $[]_i$, but
$\?_{[]_{i}}$ is not maximal for \emph{type precision} at level $i$,
so as to avoid the Fire Triangle, as explained in \cref{sec:introduction}.
Conversely, however, type precision is stable by product formations, \ie in the
non-dependent case if $A \precisionType{i} A'$ and $B \precisionType{i} B'$
then $A \to B \precisionType{i} A' \to B'$. This is not the case for term
precision, again because of the Fire Triangle and
of the maximality of $\?_{[]}$ as a term. 

This design forces certain terms to be non-monotone,
in particular those built using large elimination.
Consider the type-level function \coqe{t_0 := fun b => if b then nat else nat -> nat}. 
We have \coqe{false precise ?_bool}, but we do not have \coqe{t false equiv nat -> nat precise ?_Type0}. We can address the issue in this simple case by posing \coqe{t_1 := fun b => if b then iota nat else iota (nat -> nat)}, which explicitly uses cumulativity, so \coqe{t_1} 
is monotone as a function of type $\bool -> []_1$.
Using cumulativity however does not work for {\em recursive} large
elimination as the \coqe{nArrow} function discussed in the
introduction, \shnew{because the appropriate universe level is not
  known statically. While being typable in \CCICPrec, $\Omega$ and similar self-applications
  that would be non-terminating in \CCICG
  are also not self-precise, witnessing their pathological behavior.}

Armed with these notions of precision, it becomes possible to axiomatize
directly in $\prop$ the various properties they satisfy and their relation to
casts.
%
Note that because this axiomatization occurs in the definitionally 
proof-irrelevant universe $\prop$, there is
no need to endow the axioms with any computational meaning:
they just need to be justified by a model to guarantee consistency (\cref{sec:model}).

\subsection{Internal Reasoning about Graduality}
\label{sec:action-dgg}

Graduality~\cite{newAhmed:icfp2018} and the dynamic gradual guarantee (DGG)~\cite{siekAl:snapl2015} are usually established as global properties of a gradual language. However, as mandated by the Fire Triangle of Graduality~\cite{lennonAl:toplas2022}, graduality cannot hold globally in a terminating gradual extension of \CIC. 
While \citet{lennonAl:toplas2022} simply do not attempt to study graduality for \CCICN, the situation of \CCICPrec in this regard is both novel and unique: because precision is an internal notion within a type theory that allows for consistent reasoning, we can account for graduality. We can also exactly state the DGG theorem that holds in \CCICPrec.

%


%


\paragraph{Dynamic Gradual Guarantee}

In essence, the DGG says that if a term $x$ is more precise than a term $y$, then for any evaluation context $C$, $C~x$ ``error approximates'' $C~y$---meaning that $C~x$ can fail more than $C~y$, but if it does not fail, then both are equivalent. Essentially, this property is about the {\em monotonicity} of contexts with respect to precision.
In our setting, an evaluation context is simply a function from some type $A$ to the type $\bool$ of booleans, so the DGG corresponds to the monotonicity of functions, that is,
$
\mathtt{DGG} : \forall (A : []) (C : A -> \bool)
  (x~y : A) , 
    x \precision{A}{A} y -> C~x \precision{\bool}{\bool} C~y.
$

As we have seen above with \coqe{nArrow}, not all functions are monotone in \CCICPrec.
To establish monotonicity internally in a general manner, we need a notion that does not make sense only for function types. 
Fortunately, a direct consequence of the pointwise definition of precision on functions is that monotonicity of functions corresponds to their {\em self-precision}. In general, we write $\sp{a}{A}$ for self-precision, meaning
that $a:A$ is such that $a \precision{A}{A} a$.

In \CCICPrec, $\mathtt{DGG}~A~C$ is equivalent to $\sp{C}{A -> \bool}$. 
In other words, for any type $A$ and for any context $C$ that is self-precise, we have the usual dynamic gradual guarantee between two elements $x$ and $y$ related by the precision over $A$.
%
%
This means that we can understand existing gradual systems in which the DGG holds globally as
systems where every context is self-precise by construction.

\paragraph{Graduality}
Graduality~\cite{newAhmed:icfp2018} is defined as
the fact that when $A \precisionType{i} B$, for
any $a:A$ and $b:B$, there is an adjunction
$
\cast{A}{B}{a} \precision{B}{B} b  \leftrightarrow 
a \precision{A}{B} b \leftrightarrow 
a\precision{A}{A} \cast{B}{A}{b},
$
and furthermore the roundtrip is the identity on $A$ up to equiprecision:
$
\cast{B}{A}{\cast{A}{B}{a}} \precision{A}{A} a
$
\shnew{(the reverse precision relation 
is a consequence of reflexivity and the
  adjunction property).}


As we show in \cref{sec:precision-properties} (\cref{prop:grip-graduality}), \CCICPrec globally satisfies graduality, except for the fact that $a:A$ and $b:B$ must both be self-precise for
it to hold.

\shnew{
\paragraph{Applications}
Graduality and the DGG can be exploited in several ways using internal precision.}
A potential use is to develop internally the theory of
precision, showing for instance that casts between types related by
precision do compose (which is not the case for arbitrary types).
Another possible use is to derive proofs of precision on open
terms that can appear during reasoning. For instance, when using gradual subset types
(introduced in \cref{sec:action-subsets} below) to define functions, it becomes necessary
to discharge proof obligations related to the precision of terms containing free variables.  

\shnew{
One can also exploit the reasoning principle of the DGG for certifying gradual programs. 
We mention in \cref{sec:introduction} the case of two programs that use the $\mathtt{map}$ function and its self-precision to deduce that a gradual program does not fail. 
More generally, given any correctness criterion for $t$ (for instance that the resulting list has the same length as the input list) 
knowing $t \precision{}{} u$ is sufficient to deduce the corresponding criterion for $u$, as long as the criterion is self-precise. Considering that proofs of self-precision could be automated for a large class of terms (see \cref{thm:ccics-sp}, which in particular covers all the terms mentioned in this example), the proof burden of correctness results can be considerably lowered by exploiting the DGG compared to direct reasoning.
Alternatively, \CCICPrec lets user construct precision proofs where actual non-trivial reasoning is needed, as illustrated in the next section.}






\subsection{Exception Handling and Graduality}
\label{sec:exceptions}

All languages in the theoretical literature that address graduality 
are devoid of exception handling mechanisms. The reason is that handling runtime type errors makes it possible to define  terms that are not monotone with respect to precision, and so graduality cannot hold globally. However, in practice, exception handling (and other language mechanisms in tension with graduality) are key ingredients and one would ideally like to account for them.
As explained above, the situation of \CCICPrec in this regard is new and singular: since we can internally and consistently reason about precision, we can support exception handling terms, and still establish their monotonicity as specific theorems proven in the type theory itself. Below we illustrate such an exception-handling term and its proof of monotonicity within \CCICPrec.

The catch operator on $\bool$ is not monotone with respect to precision. Consider its type signature: 
$$
\bcatch : \forall (A: []) \ (a_\btrue : A) \ (a_\bfalse : A) \
(a_{\err_\bool} : A)\ (a_{\?_\bool} : A) , \bool -> A
$$
\shnew{There is no reason for $a_{\?_\bool}$, given to handle the unknown term case, to be less
precise than $a_\btrue$ and $a_\bfalse$.}
In our setting, the catch operation (and its dependent generalization)
can be considered, without endangering any properties of the system.
Moreover, we can show  that precision is preserved in specific uses of catch.

To illustrate, consider the following optimized implementation of
(iterated) multiplication of a list of natural numbers, with two functions, that takes advantage of the
fact that $0$ is an absorbing element (we use pattern matching syntax
for induction on lists to ease the reading):
$$
\begin{array}{llcl}
  \mathtt{mult}^{\err}_\listT& \nilK & := & 1 \\
  \mathtt{mult}^{\err}_\listT & (\consK~n~l) & := &
                                                \ifte{(\mathtt{is\_zero}~n)}{\err_\nat}
                                                {n * \mathtt{mult}^{\err}_\listT l}\\[0.5em]
  \mathtt{mult}_\listT & l & :=  & \ncatch~\nat~0~(\l n:\nat . 1 + n)~0~\?_\nat~
                               (\mathtt{mult}^{\err}_\listT~l)                                                
\end{array}
$$
The function $\mathtt{mult}^{\err}_\listT$ returns an error as soon as
a $0$ is encountered in the list, short-circuiting the recursive computation.
The wrapper function $\mathtt{mult}_\listT$ catches
errors raised by $\mathtt{mult}^{\err}_\listT$ and returns $0$ in
that case.
%
%
In general, $\mathtt{mult}_\listT$ is not monotone because when the
input list is an error, it returns the value $0$, which is not more
precise than the return value on other lists.
But $\mathtt{mult}_\listT$ is monotone on lists that do
not contain errors, because in such cases errors are used
in a delimited manner in order to optimize execution.
In \CCICPrec, we can make this explicit and prove the
following theorem:
$$
\mathtt{mult}_\listT^\sqsubseteq : \forall (l \ l': \listT~\nat),
\mathtt{not}\mbox{-}\err_{\listT}~ l ->
l \precision{\listT~\nat}{\listT~\nat} l' -> 
\mathtt{mult}_\listT~l \precision{\nat}{\nat} \mathtt{mult}_\listT~l'.
$$
where $\mathtt{not}\mbox{-}\err_{\listT}$ is a predicate ensuring that
the list is not $\err_{\listT~\nat}$ and does not contain
$\err_\nat$ in its elements.
\shnew{Again, details can be found in the Agda development.}

\subsection{Gradual Subset Types}
\label{sec:action-subsets}


The logical layer $\prop$ enables stating and proving formal properties
on the gradual, impure layer $[]$. But in a dependently-typed setting, it is also
important to be able to use the properties stated in $\prop$
to constrain types in $[]$, using for instance 
{\em subset types}. Recall that a subset type is a type $A$ enriched with a proposition $P$, 
noted $\{ a : A \ \& \ P~a \}$, and an inhabitant is a dependent pair $(a ; p)$, such that $a : A$ and $p : P~a$.
This means that in \CCICPrec we need a way to embed $\prop$ into $[]$.
Note that this cannot be a direct injection, as propositions in
$\prop$ cannot be inhabited with exceptions.
Therefore, we need a special operator $\BoxProp : \prop -> []$ that
takes a proposition $P$ and freely adds $\err_{\BoxProp~P}$ and
$?_{\BoxProp~P}$ to $P$.
This allows us to define lists of size $n$ as the type
$$
\sizedlist~A~n := \{ l : \listT~A \ \& \ \BoxProp~(\len~l = n) \}.
$$
This way, we can gradually define the $\mathtt{append}_?$ function as
$$
\begin{array}{l}
  \mathtt{append}_? : \forall A \ n \ m , \ \sizedlist~A~n ->
  \sizedlist~A~m ->
  \sizedlist~A~(n+m) \\
\mathtt{append}_? \ A  \ n \ m \ (l ; \_) \ (l' ; \_) := ( l \append l' ; ?_{\BoxProp~(\len~
  (l \append l') = n + m)})
\end{array}
$$
where the proof that the result is of the right size is avoided through imprecision.
It is also possible to define the precise append function that
contains the actual proof that the resulting size is valid:
$$
\begin{array}{l}
  \mathtt{append} : \forall A \ n \ m , \ \sizedlist~A~n ->
  \sizedlist~A~m ->
  \sizedlist~A~(n+m) \\
\mathtt{append} \ A  \ n \ m \ (l ; \boxK~p) \ (l' ; \boxK~p') := ( l
  \append l' ; \append\!\mathtt{lemma} \ l \ l' \cdot \mathtt{ap}_2 + p \ p' )
\end{array}
$$
where $\append\!\mathtt{lemma}$ is the proof that the length of two
appended lists is equal to the sum of their lengths,
$e \cdot e'$ is the concatenation of equality and $\mathtt{ap}_2$ is a
witness that (binary) functions preserve equalities. 

%
In \CCICPrec, these two append functions can be distinguished in the logical
layer by using the following predicate, which indicates that a
property in the impure layer has {\em really} been proven: 
$$
\begin{array}{ll}
\mathtt{valid}_{\BoxProp} : \forall P:\prop, \ \BoxProp~P -> \prop  & 
\mathtt{valid}_{\BoxProp} \ P \ (\err_{\BoxProp~P}) := \botProp\\
\mathtt{valid}_{\BoxProp} \ P \ (\boxK~p) := \topProp &
\mathtt{valid}_{\BoxProp} \ P \ (?_{\BoxProp~P}) := \botProp
\end{array}
$$
Posing $\mathtt{valid}_{\sizedlist} (\_ ; p) :=
\mathtt{valid}_{\BoxProp} \ \_ \ p$, 
the precise append function is the only one of the two
versions for which one can prove:
$$
\mathtt{valid}\_\mathtt{append} : \forall \ A \ n \ m \ l \ l' ,
\mathtt{valid}_{\sizedlist} l ->
\mathtt{valid}_{\sizedlist} l' ->
\mathtt{valid}_{\sizedlist} (\mathtt{append} \ A  \ n \ m \
l \ l')
$$

In a gradual setting, we can also use the unknown term in order to avoid an
explicit definition of the resulting size of the list.
For instance, the filter function can be given the imprecise type
$$
\mathtt{filter} : \forall A \ n \ (P : A -> \prop) , \ \sizedlist~A~n ->
  \sizedlist~A~?_{\nat}
$$
However, there is no way to give a valid implementation of a filter
function of that type, because the size of the filtered list cannot
be proven to be equal to $?_{\nat}$ in the logical layer.
Taking advantage of the internal notion of precision, we can define an
alternative notion of sized list in \CCICPrec as
$$
\sizedlistprec~A~n := \{ l : \listT~A \ \& \ \BoxProp~(\len~l \precision{\nat}{\nat} n) \}.
$$
Using this notion of sized lists, it is possible to
define a valid filter function of type
$$
\mathtt{filter}_{\precision{}{}} : \forall A \ n \ (P : A -> \prop) , \ \sizedlistprec~A~n ->
  \sizedlistprec~A~?_{\nat}.
$$
because the proof that the size of the filtered list is more precise
than $?_{\nat}$ directly follows from the fact that $?_{\nat}$ is
the maximal element of type $\nat$.

\section{Gradual types and pure propositions}
\label{sec:ccic}

In this section, we present the two-layer core of \CCICPrec,
intended to be both a gradual cast calculus, target for elaboration
of a gradual surface language, and a pure language to consistently 
reason about programs in that cast calculus.
In \cref{sec:ccic-grad}, we give an overview of the gradual
part of the language, while \cref{sec:ccic-prop} introduces
the pure sort of propositions. Finally, \cref{sec:ccic-crossing}
discusses how to soundly support interactions between these two layers. 

\subsection{The Impure Layer of Gradual Terms}
\label{sec:ccic-grad}

\begin{figure}
\begin{small}
  \boxedrule{$\vdash \Gamma$}
  \vspace{-.5em}
	\begin{mathpar}
  \vdash \cdot \quad \text{and} \quad 
  \vdash \Gamma, x : T \text{ whenever }
  \vdash \Gamma \text{ and }
      \Gamma \vdash T \ccicty []_{i}
	\end{mathpar}
	\boxedrule{$\Gamma \vdash t \ccicty T$} \vspace{1em}
	\begin{mathpar}
		\inferrule[Univ]
    {\vdash \Gamma}
    {\Gamma \vdash []_{i} \ccicty []_{i + 1}}
    \ilabel{infrule:cic-univ}
    \and
		\inferrule[Var]
    {\vdash \Gamma \\ (x : T) \in \Gamma}
    {\Gamma \vdash x \ccicty T}
    \ilabel{infrule:cic-var}
    \and
		\inferrule[Prod]
    {\Gamma \vdash A \ccicty []_i \\
      \Gamma, x : A \vdash B \ccicty []_{i}}
    {\Gamma \vdash  \P x : A . B \ccicty []_{i}}
    \ilabel{infrule:cic-prod}
    \\
		\inferrule[Abs]
    { \Gamma \vdash A \ccicty []_{i} \\
      \Gamma, x : A \vdash t \ccicty B}
    {\Gamma \vdash \l x : A . t \ccicty  \P x : A . B}
    \ilabel{infrule:cic-abs}
    \and
		\inferrule[App]
    { \Gamma \vdash t \ccicty  \P x : A . B \\
      \Gamma \vdash u \ccicty A}
    {\Gamma \vdash t~u \ccicty B\subs{x}{u}}
    \ilabel{infrule:cic-app}
    \\
    \inferrule[List]
    {\Gamma{} \vdash A : []_i}
    {\Gamma{} \vdash \listT\,A : []_i}
    \ilabel{infrule:ccic-list}
    \and
    \inferrule[List-Nil]
    {\Gamma{} \vdash A : []_i}
    {\Gamma{} \vdash \nilK[A] : \listT\,A}
    \ilabel{infrule:ccic-nilK}
    \and
    \inferrule[List-Cons]
    { \Gamma \vdash A : []_i \\
      \Gamma{} \vdash a : A \\
      \Gamma{} \vdash l : \listT\,A}
    {\Gamma{} \vdash \consK[A]\,a\,l : \listT\,A}
    \ilabel{infrule:ccic-consK}
    \and
    \inferrule[List-Catch]
    { \Gamma \vdash P : \listT\,A \to []_i\\
      \Gamma \vdash h_{\nilK} : P\,\nilK[A]\\
      \Gamma \vdash h_{\consK} : \P (a : A)(l : \listT\,A) . P\,l {\to} P(\consK[A]\,a\,l)\\
      \Gamma \vdash h_{\err} : P\,\err_{\listT~A}\\
      \Gamma \vdash h_{\?} : P\,\?_{\listT~A}\\
      \Gamma \vdash l : \listT\,A \\
    }
    {\Gamma \vdash \lcatch{A}\,P\,h_{\nilK}\,h_{\consK}\,h_{\err}\,h_\?\,l : P\,l}
    \ilabel{infrule:ccic-list-ind}
    \\
    \inferrule[Unk]
    {{\Gamma} \vdash {A} \ccicty {[]_i}}
    {{\Gamma} \vdash {\?_A} \ccicty {A}}
    \ilabel{infrule:ccic-unk}
    \and
    \inferrule[Err]
    {{\Gamma} \vdash {A} \ccicty {[]_i}}
    {{\Gamma} \vdash {\err_A} \ccicty {A}}
    \ilabel{infrule:ccic-err}
    \and
    \inferrule[Cast]
    { {\Gamma} \vdash {A} \ccicty {[]_i} \\
      {\Gamma} \vdash {B} \ccicty {[]_i} \\
      {\Gamma} \vdash {t} \ccicty {A}}
    {{\Gamma} \vdash {\cast{A}{B}{t}} \ccicty {B}}
    \ilabel{infrule:ccic-cast}
    \\
    \inferrule[Cum]
    { \Gamma \vdash A \ccicty []_i}
    {\Gamma \vdash \cum A \ccicty []_{i+1}}
    \ilabel{infrule:ccic-cum}
    \and
    \inferrule[Coe]
    { \Gamma{} \vdash a : A }
    { \Gamma{} \vdash \coe\,a : \cum A }
    \ilabel{infrule:ccic-coe}
    \and
    \inferrule[Coe-Inv]
    { \Gamma{} \vdash a : \cum A }
    { \Gamma{} \vdash \coeinv\,a : A }
    \ilabel{infrule:ccic-coeinv}

		\inferrule[Conv]
    { \Gamma \vdash t \ccicty T' \\
      \Gamma \vdash T : []_i \\
      \Gamma \vdash T' \conv T : []_i}
    {\Gamma \vdash t \ccicty T}
    \ilabel{infrule:ccic-conv}
		
	\end{mathpar}

  \boxedrule{$\Gamma \vdash t \conv t' : T$} \hspace{3mm} (Congruence, reflexivity, symmetry and transitivity rules omitted)
  
  \begin{mathpar}
    %
    \inferrule[Conv-Red]
    {\Gamma \vdash t: A \\ \Gamma \vdash t' : A \\ t \redCCIC t' }
    {\Gamma{} \vdash t \conv t' : A}
    \ilabel{infrule:ccic-conv-red}
    \and
    \inferrule[$\P$-$\eta$]
    {\Gamma{}, x : A \vdash t\,x \conv t'\,x : B}
    {\Gamma{} \vdash t \conv t' : \P x : A. B}
    \ilabel{infrule:ccic-pi-eta}
    \\
    \inferrule[Coe-Retr]
    {\Gamma{} \vdash a : A}
    {\Gamma{} \vdash \coeinv\,(\coe\,a) \conv a : A}
    \ilabel{infrule:coe-retr}
    \and
    \inferrule[Coe-Sect]
    {\Gamma{} \vdash a : \cum\,A}
    {\Gamma{} \vdash \coe\,(\coeinv\,a) \conv a : \cum\,A}
    \ilabel{infrule:coe-sect}
  \end{mathpar}
\end{small}

	\caption{\CCICPrec: typing of the impure layer --- based on \CCIC}
	\label{fig:ccic-base}
\end{figure}
As seen in \cref{sec:ccic-background}, 
\CCIC is an extension of \MLTT with primitives for gradual typing, namely casts, errors and unknown terms.
For the impure layer of gradual terms, \CCICPrec follows significantly \CCIC~\cite{lennonAl:toplas2022}, with some minor modifications and presentation differences highlighted below, in particular the support for exception handling and explicit cumulativity.

The syntax and typing rules of the gradual layer of \CCICPrec are given in~\cref{fig:ccic-base}.
They feature a hierarchy of universes $[]_i$, dependent products $\Pi$
introduced by $\lambda$-abstraction and destructed by applications, and
inductive types, introduced by constructors
and destructed by catch operators.
Here
we do not consider inductive types with indices,
such as equality, whose treatment is deferred to~\cref{sec:equality}.
For readability we only formally present lists $\listT$,
however the calculus can readily be extended with other parametrized instances
of $\mathbb{W}$-types (see \cref{sec:general-inductive-types}), as done for \CCIC.
Throughout the article, and in particular for examples, we take
the liberty to use dependent sums $\Sigma$, natural numbers $\nat$ and booleans $\bool$.
The typing rule (\nameref{infrule:ccic-list-ind}) for the catch
operator on lists requires two additional arguments with respect to
the usual recursor on lists, one for the case of an error, and one for $\?$.
Note that the usual recursor on lists $\lind{}$ which simply propagates $\err$
and $\?$, as used in \CCIC,
can be recovered from the catch operator by defining $h_{\err}$ to be
$\err$ and $h_\?$ to be $\?$.

Like \Agda, \CCICPrec uses explicit cumulativity. The operator $\cum$ lifts a type
from one universe to the next, and operators $\coe$ and $\coeinv$ coerce between
a type and its lift.
%
%
We choose explicit 
cumulativity due to the central role it plays in the definition of internal
precision (\cref{sec:precision}---see \cref{sec:impl-deta-agda} for further discussion on explicit versus
implicit cumulativity). 
As for the gradual part of the calculus, it
features the unknown terms $\?_{A} : A$,
errors $\err_{A} : A$, and casts $\ascdom{a}{A}{B}$ between arbitrary types
at the same universe level.

As any dependent type theory, \CCICPrec relies on a notion of conversion that
allows us to convert a term of type $T'$ to a term of type $T$  (Rule
\nameref{infrule:ccic-conv}) as soon the two types are convertible.
Conversion is defined as the reflexive, symmetric and transitive
closure of reduction with the additional $\eta$-conversion for
functions and the fact that $\coe$ and $\coeinv$ are inverse of each
other. 
The dynamic behavior of these terms is presented
by means of a reduction relation in~\cref{fig:ccic-reduction}, directly adapted from that of \CCIC.
There are three sets of rules. The first is for standard rules of
\MLTT, \ie the usual $\beta$-rule
for functions and $\iota$-rule for lists.
The second corresponds to propagation of both $\?$ and $\err$ as
exceptions as advocated for by \citet{pedrotTabareau:esop2018}.
The last describes the behavior of the cast
primitive, which computes based on the shape of its two type arguments.
The first five rules propagate casts between types with the same head constructor.
The next four correspond to failures, either when the source and target
types are incompatible, when one of them is an error, or when trying to
cast a product type into the unknown type of its level.
This last rule~\nameref{redrule:prod-err} is crucial for normalization,
as it is responsible for the failure of terms such as $\Omega$.
Next, rule~\nameref{redrule:up-down} can be understood as a form of $\iota$-rule for $\?_{[]}$: it showcases the fact that casts into $\?_{[]}$
work as canonical forms for it (when their domain is of a certain form),
with casts from $\?_{[]}$ as destructors.
Finally rule~\nameref{redrule:ind-germ} decomposes casts from a list
into the unknown type through $\listT \?_{[]}$,
the most general type with $\listT$ as a head constructor,
letting rules~\nameref{redrule:l-l-nil} and \nameref{redrule:l-l-cons}
further decompose the innermost cast if applicable.
\shnew{Finally, \nameref{infrule:red-cong} complements the top-level reduction given by the other rules with congruence closure. For the
purpose of that rule $\listT$, $\consK[A]$, $\nilK[A]$ and $\lcatch{A}$ are treated as terms applied to their arguments.}

\shnew{As standard in rewriting systems for programming languages, reduction is orthogonal (left-linear and without critical pairs), and so
the standard parallel reduction proof technique \cite{TAKAHASHI1995120} applies to show that it is confluent. This is further witnessed by the
confluence checker of Agda, which accepts the definitions of the proof-of-concept implementation.}

\begin{figure}
\begin{small}
  \boxedrule{$\H$, $\types$, $\hd : \types \to \H$
  }\vspace{2mm}

  \begin{mathpar}

  \H \ni h ::= [] \mid \Pi \mid \listT \mid \cum A \and

  \types ::= []_i \mid \P x : A . B \mid \listT\,A \mid \cum A \\

  \hd([]_i) := []  \and
  \hd(\P x : A . B) := \Pi  \and
  \hd(\listT\,A) := \listT  \and
  \hd(\cum A) := \cum
  \end{mathpar}

  \boxedrule{$t \redCCIC t'$}\\
  \begin{mathpar}
    \redruleAux{(\l x: A . t)~u}
    {t \subs{u}{x}
    }{$\P$-$\beta$}
    \ilabel{redrule:prod-beta}\and
    \redruleAux{
      \lcatch{A}\,P\,h_{\nilK}\,h_{\consK}\,h_{\err}\,h_\?\,\nilK[A]
    }{
      h_{\nilK}
    }{Catch-$\nilK$}
    \ilabel{redrule:lind-nil} \and
    \redrule{
      \lcatch{A}\,P\,h_{\nilK}\,h_{\consK}\,h_{\err}\,h_\?\,(\consK[A]\,a\,l)
    }{
      h_{\consK}\,a\,l\,(\lcatch{A}\,P\,h_{\nilK}\,h_{\consK}\,h_{\err}\,h_\?\,l)
    }[Catch-$\consK$]
    \ilabel{redrule:lind-cons}\\
    
    \text{\textbf{Propagation rules for $\?$ and $\err$}} \\
    \redruleAux{
      \?_{\P(x:A).B}
    }{
      \l (x : A). \?_{B}}{$\Pi$-Unk}
    \ilabel{redrule:prod-unk} \and
    \redruleAux{
      \err_{\P(x:A).B}
    }{
      \l (x : A). \err_{B}}{$\Pi$-Err}
    \ilabel{redrule:prod-err} \\
  \redruleAux{
    \?_{\cum A}
  }{
    \coe\,\?_{A}
  }{Cum-Unk} \ilabel{redrule:cum-ukn} \and
  \redruleAux{
    \lcatch{A}\,P\,h_{\nilK}\,h_{\consK}\,h_{\err}\,h_\?\,\?_{\listT\,A}
  }{
    h_\?
  }{Catch-Unk} \ilabel{redrule:lind-unk} \and
  \redruleAux{
    \err_{\cum A}
  }{
    \coe\,\err_{A}
  }{Cum-Err} \ilabel{redrule:cum-err} \and
  \redruleAux{
    \lcatch{A}\,P\,h_{\nilK}\,h_{\consK}\,h_{\err}\,h_\?\,\err_{\listT\,A}
  }{
    h_{\err}
  }{Catch-Err} \ilabel{redrule:lind-err} \and
  \redruleAux{
    \ascdom{\?_{\listT\,A}}{\listT\,A''}{\listT\,A'}
  }{
    \?_{\listT\,A'}
  }{$\listT$-Cast-Unk} \ilabel{redrule:lcast-unk} \and
  \redruleAux{
    \ascdom{\err_{\listT\,A}}{\listT\,A''}{\listT\,A'}
  }{
    \err_{\listT\,A'}
  }{$\listT$-Cast-Err} \ilabel{redrule:lcast-err} \and
  \redruleAux{
    \ascdom{\?_{\?_{{[]_i}}}}{\?_{{[]_i}}}{T}
  }{
    \?_{T}
  }{Down-Unk} \ilabel{redrule:down-unk} \and
    \and 
  \redruleAux{
    \ascdom{\err_{\?_{{[]_i}}}}{\?_{{[]_i}}}{T}
  }{
    \err_{T}
  }{Down-Err} \ilabel{redrule:down-err} 
  \text{ when $T \in \types$} \\
  \text{\textbf{Reduction rules for cast}} \\

  \redrule{
    	\ascdom{f}{\P (x: A_1). B_1}{\P (y:A_2).B_2}
    }{
      \l y : A_2. \ascdom{(f~\cast{A_2}{A_1}{y})}{B_1\subs{\cast{A_2}{A_1}{y}}{x}}{B_2} 
    }[$\Pi$-$\Pi$] \ilabel{redrule:prod-prod} \\
%
  \redruleAux{
    \cast{\cum\,A}{\cum\,A'}{\coe\,t}
  }{
    \coe\,\cast{A}{A'}{t}
  }{Cum-Cum} \and
  \redruleAux{
    \ascdom{A}{[]_i}{[]_i}
  }{A}{Univ-Univ} \ilabel{redrule:univ-univ} \and
  \redrule{
    \ascdom{\nilK[A]}{\listT\,A''}{\listT\,A'}
  }{
    \nilK[A']
  }[$\listT$-$\listT$-Nil] \ilabel{redrule:l-l-nil} \and
  \redrule{
    \ascdom{\left(\consK[A]\,a\,l\right)}{\listT\,A''}{\listT\,A'}
  }{
    \consK[A']\,\left(\cast{A}{A'}{a}\right)\,\left(\cast{\listT\,A}{\listT\,A'}{l}\right)
  }[$\listT$-$\listT$-Cons] \ilabel{redrule:l-l-cons} \and
	\redrule{
    \ascdom{t}{T}{T'}
  }{
    \err_{T'}
  }[Head-Err] \ilabel{redrule:head-err}
	  \hfill \text{when $T, T' \in \types$ and $\hd T \neq \hd T'$} \\
  \redruleAux{
    \ascdom{t}{\err_{{[]_i}}}{T}
  }{
    \err_T
  }{Dom-Err} \ilabel{redrule:err-dom} \and
%
  \redruleAux{
    \ascdom{t}{T}{\err_{{[]_i}}}
  }{
    \err_{\err_{[]_i}}
  }{Cod-Err} \ilabel{redrule:err-codom} \and
  \text{when $T \in \types$} \\
  \redrule{
    \ascdom{f}{\P x:A.B}{\?_{[]_i}}
  }{
    \err_{\?_{[]_i}}
  }[Cast-$\Pi$-Err] \ilabel{redrule:cast-prod-err} \\
%
%
%
  \redrule{
    \cast{\?_{[]_i}}{Y}{\cast{X}{\?_{[]_i}}{t}}
	}{
		\ascdom{t}{Y}{X}
  }[Up-Down] \ilabel{redrule:up-down}
	\hfill \text{when $X \in \types$ and $Y$ is $\listT\,\?_{[]}$, $[]$ or $\cum A$} \\
  \redrule{
    \ascdom{t}{\listT\,A}{\?_{[]_i}}
  }{
    \castmid{\listT\,A}{\listT\,\?_{[]_i}}{\?_{[]_i}}{t}
  }[$\listT$-Dec] \ilabel{redrule:ind-germ}
  \hfill \text{when $\listT\,A \neq \listT\,\?_{[]}$}\\

  \text{\textbf{Congruence} ($A$, $B$ and $t$ denote arbitrary terms)}
\end{mathpar}

\begin{minipage}{0.75\textwidth}
  \begin{align*}
    \Ctx ::= &~ [\cdot] \mid \P x : \Ctx.\ B \mid \P x : A. \Ctx \mid  \l x : \Ctx.\ t \mid \l x : A.\ \Ctx \mid t\ \Ctx \mid \Ctx\ t
      \\
      & \mid \?_{\Ctx} \mid \err_{\Ctx} 
      \mid \cast{\Ctx}{B}{t} \mid \cast{A}{\Ctx}{t} \mid \cast{A}{B}{\Ctx} \mid \cum\ \Ctx \mid \coe\ \Ctx \mid \coeinv\ \Ctx \\
      & \mid \nilK[\Ctx] \mid \consK[\Ctx] \mid \lcatch{\Ctx}
    \end{align*}
\end{minipage}%
\begin{minipage}{0.25\textwidth}
  \[\inferrule[Red-Cong]{t \redCCIC t'}{\Ctx[ t ] \redCCIC \Ctx[ t' ]}     \ilabel{infrule:red-cong}  \]
\end{minipage}
\vspace{1em}
\end{small}

  \caption{\CCICPrec: Reduction rules -- adapted from \CCIC}
  \label{fig:ccic-reduction}
\end{figure}

\subsection{The Pure Layer for Reasoning on Gradual Terms}
\label{sec:ccic-prop}

The casts $\ascdom{t}{A}{B}$ and exceptional terms $\err_A, \?_A$
are fundamental features to enable gradual programming.
However, as a consequence all types are inhabited, so logical consistency, and
thus meaningful internal reasoning on programs, is lost.
To remedy this problem, following the insight of \RETT \cite{pedrotAl:icfp2019},
we introduce an additional layer dedicated to sound reasoning,
which must therefore be free of the gradual primitives.
As in \RETT, the separation between the impure and pure layers is
controlled by means of sorts: alongside the impure hierarchy of gradual terms $[]_i$, we
introduce a new impredicative\footnote{
  Impredicativity is inessential but simplifies the
  exposition while matching the model in~\cref{sec:model}; the Agda development
  shows how this presentation can be adapted to a predicative
  hierarchy $\prop_i$.} sort $\prop$ of definitionally
proof-irrelevant pure propositions.
\shnew{Since the propositional layer is pure, there is no “unknown proposition” $?_{P}$ for a proposition $P$.
But this is not needed, because in that layer axioms suffice, as they are readily convertible to any other term by
propositional irrelevance.}

In more details, \cref{fig:ccic-prop} shows how \CCICPrec extends
what was essentially \CCIC with this new sort $\prop$ (\nameref{infrule:prop-wf}). In particular, 
an extension of conversion specifies that any two proofs of the same
proposition are convertible (\nameref{infrule:prop-irr}).
We use $\sort$ for a generic sort,
that is either $\prop$ or $[]_i$ for some $i$.
At this stage, there are only two ways to construct propositions.
On one side, the empty proposition $\botProp$
(\nameref{infrule:bot-prop-wf}) with no introduction, and elimination in the form
of an explosion principle
(\nameref{infrule:bot-prop-elim}).
On the other, universal quantification over propositions or
types (\nameref{infrule:forall-wf}) introduced by $\l$-abstraction
(\nameref{infrule:forall-intro}) and eliminated by application
(\nameref{infrule:forall-elim}).
Implication $P \to Q$ between propositions is defined as the non-dependent
quantification $\forall (\_ : P), Q$.
More interesting ones will be added later, such as the
precision relation (\cref{fig:ccic-prec}).
However, further logical connectives can already be encoded on top of the
primitives we already have, using impredicativity and definitional
proof-irrelevance \cite{gilbert:hal-01859964}. For instance, the
proposition true can defined by $\topProp := \botProp -> \botProp$.

\begin{figure}
  \flushleft{}
\begin{small}
  \begin{mathpar}
    \inferrule[$\prop$-Wf]
    {\Gamma{} \vdash}
    {\Gamma{} \vdash \prop : []_{0}}
    \ilabel{infrule:prop-wf}
    \and
    \inferrule[$\prop$-Irr]
    {\Gamma \vdash P : \prop\\\Gamma\vdash p,q : P}
    {\Gamma \vdash p \conv q : P}
    \ilabel{infrule:prop-irr}
    \and
    \inferrule[$\botProp$-Wf]{\Gamma{} \vdash}{\Gamma{} \vdash \botProp : \prop}
    \ilabel{infrule:bot-prop-wf}
    \and
    \inferrule[$\botProp$-Elim]
    {\Gamma{} \vdash p : \botProp\\ \Gamma{} \vdash A : \sort}
    {\Gamma{} \vdash \exfalso{p}{A} : A}
    \ilabel{infrule:bot-prop-elim}
    \\
    \inferrule[$\forall$-Wf]
    {\Gamma{} \vdash A : \sort \\ \Gamma{}, x : A \vdash P : \prop}
    {\Gamma{} \vdash \forall (x : A), P : \prop}
    \ilabel{infrule:forall-wf}
    \and
    \inferrule[$\forall$-Intro]{\Gamma{}, x : A \vdash p : P}{\Gamma{} \vdash
      \lambda (x : A). p : \forall (x : A), P}
    \ilabel{infrule:forall-intro}
    \and
    \inferrule[$\forall$-Elim]{
      \Gamma{} \vdash f : \forall(x : A), P\\
      \Gamma{} \vdash a : A}
    {\Gamma{} \vdash f~a : P[a/x]}
    \ilabel{infrule:forall-elim} \\

    \inferrule[List-Catch-Prop]
    { \Gamma \vdash P : \listT\,A \to \prop\\
      \Gamma \vdash h_{\nilK} : P\,\nilK[A]\\
      \Gamma \vdash h_{\consK} : \P (a : A)(l : \listT\,A) . P\,l {\to} P(\consK[A]\,a\,l)\\
      \Gamma \vdash h_{\err} : P\,\err_{\listT~A}\\
      \Gamma \vdash h_{\?} : P\,\?_{\listT~A}\\
      \Gamma \vdash l : \listT\,A \\
    }
    {\Gamma \vdash \lcatchProp{A}\,P\,h_{\nilK}\,h_{\consK}\,h_{\err}\,h_\?\,l : P\,l}
    \ilabel{infrule:ccic-list-ind-prop}
    \\       
    \inferrule[$\BoxProp$-Wf]{\Gamma{} \vdash P : \prop}
    {\Gamma{} \vdash \BoxProp~P : []_0}
    \ilabel{infrule:ccic-box}
    \and
    \inferrule[$\BoxProp$-Intro]
    {\Gamma{} \vdash p : P}
    {\Gamma{} \vdash \boxK_{P}~p : \BoxProp~P}
    \ilabel{infrule:ccic-box-intro}
    \vspace{-3mm}
    \and
    \inferrule[$\BoxProp$-Elim]
    { \Gamma \vdash A : \BoxProp\, P \to \sort \quad
      \Gamma \vdash h : \P (p : P) . A(\boxK_{P} p)\\
      \Gamma \vdash h_{\err} : A\,\err_{\BoxProp\, P}\\
      \Gamma \vdash h_{\?} : A\,\?_{\BoxProp\, P}\\
      \Gamma \vdash t : \BoxProp\, P \\
    }
    {\Gamma \vdash \boxcatch{P} A\,h\,h_{\err}\,h_\?\,l : A\,t}
    
    \ilabel{infrule:ccic-box-elim}
    \\
    \redruleAux{\ascdom{t}{\BoxProp\,P}{\BoxProp\,Q}}{\err_{\BoxProp\,Q}}{Box-Box}
    \ilabel{redrule:box-box}

    \Ctx ::= \ldots 
   \mid \forall (x : \Ctx),B \mid \forall (x : A),B \mid \boxK_{\Ctx} 
    \mid \boxcatch{\Ctx}
    \end{mathpar}

\end{small}
  \caption{\CCICPrec: Extensions of typing and reduction for propositions and boxing ($\sort = \prop$ or $[]_i$)}
  \label{fig:ccic-prop}
\end{figure}
The success of the separation of layers is given by the following theorem, proven in~\cref{sec:model}.
\begin{theorem}[Logical soundness of \CCICPrec]
  \label{thm:logical-soundness}
  If \MLTT extended with strict propositions is consistent then there is no
  closed proof $\vdash e : \botProp$ of the empty proposition $\botProp: \prop$ in
  \CCICPrec.
\end{theorem}

\subsection{Crossing Sort Boundaries}
\label{sec:ccic-crossing}

\paragraph{Eliminations}
Because of the important differences between the two layers of \CCICPrec, their
interactions need to be finely controlled in order to
stay well-behaved.
This is done by providing restricted elimination of inhabitants of
types from one layer to types of the other.

In one direction, eliminating from the pure propositional layer to the impure
gradual one is allowed only through the empty proposition $\botProp$,
by using the explosion principle, a.k.a.~{\em ex-falso} (\nameref{infrule:bot-prop-elim}).
This can be seen as a strengthening of the singleton elimination criterion
of the usual \coqe{Prop} sort of Coq,
in a way that respects definitional proof-irrelevance~\cite{gilbert:hal-01859964}.
Effectively, one is allowed to use a proof of a proposition to inhabit
a type only to show that we are in an inconsistent context,
typically in an unreachable branch of a match.
In practice, this ends up not being too restrictive,
since quite a few propositions are defined on top of $\botProp$.
For instance, internal precision defined in \cref{sec:precision} ultimately reduces to a combination of $\forall$ and
$\botProp$ after case analysis on its type parameters.
In the other direction,
eliminators from the impure layer to the pure layer
need to take errors and $\?$ into account.
Indeed, since these terms do not exist as
propositions, they cannot be used when matching on an impure argument.
Thus, the need for a $\catch$ recursor is even more dire than for
types, because we cannot rely on errors in the target type to provide
``default'' values for an $\err$ or $\?$ scrutinee,
as an $\ind$ recursor does.
On lists, for instance, we get $\lcatchProp{A}$, which behaves exactly
the same as $\lcatch{A}$ except that it can be used on predicates of type
$\listT\,A \to \prop$. 


\paragraph{Embedding Propositional Invariants within $[]$}
In order to quantify over a proposition in a type, or carry a proof along some
data, propositions must be embeddable into types and equipped with $\err$ and
$\?$.
As illustrated in \cref{sec:action-subsets} with the case of gradual subset types, 
this is achieved through the type $\BoxProp\,P$ (\cref{fig:ccic-prop})
that packs a proposition $P : \prop$ (\nameref{infrule:ccic-box}).
A proof $p : P$ of a proposition can be used to inhabit $\BoxProp\,P$ using the
constructor $\boxK_P$ (\nameref{infrule:ccic-box-intro}).
Moreover, as any other type, $\BoxProp\,P$ is equipped with exceptional
constructors $\err_{\BoxProp\,P}$ and $\?_{\BoxProp\,P}$.
The eliminator on $\BoxProp$ is given by a catch operator, similar to
the one for lists (\nameref{infrule:ccic-box-elim}), whose obvious
reduction rules are omitted.

We extend the reduction of casts to $\BoxProp$ (\nameref{redrule:box-box}) by
reducing a cast between $\BoxProp$-types to an error.
This peculiar definition is chiefly due to the fact that we cannot decide
entailment between arbitrary propositions $P$ and $Q$, and so cannot decide
when casting $\boxK_{P}~p$ to $\BoxProp~Q$ should return some $\boxK_{Q}~p'$ 
or fail.

\section{Internalizing precision}
\label{sec:precision}

The pure logical layer $\prop$ is used to assert properties of the impure gradual layer $[]$.
But none of the primitives introduced in \cref{sec:ccic} enable direct reasoning
on the most important relation between gradual programs: precision.
In this section, we provide exactly this,
by extending the logical layer with an internal precision
relation specifying the behavior of casts (\cref{sec:prec-relation}).

However, having a definition of precision is not enough:
as we cannot reason by induction
on types, general properties such as transitivity of precision cannot
be derived from the definition in \cref{sec:prec-relation} alone.
This is why we also need to directly add properties of precision
(\cref{sec:precision-properties}).
As those are added as new constants inhabiting propositions,
we do not need to specify anything about them. Indeed,
all inhabitants of propositions are definitionally equal, so none of them
is better than another. The only thing of importance is to preserve consistency of
the theory, by ensuring that the properties are validated by the model (\cref{sec:model}).

Although the impure layer does not globally satisfy graduality, a large fragment of the language
 behaves well, in the sense that it is monotone
with respect to precision (\cref{sec:prec-monotone-frag}). In particular, we show that this
fragment subsumes \shnew{
  \CCICPrecs, a fragment inspired by} \CCICs, the normalizing gradual variant proposed by \cite{lennonAl:toplas2022}
(\cref{thm:dgg-frag}).


\subsection{The Precision Relation}
\label{sec:prec-relation}
The \emph{raison d'être} of the propositional layer is to host the precision
relation, that provides
an entry point for specifying correctness properties of casts.
Precision is formulated in two distinct flavors for types and terms:
a homogeneous relation $A \precisionType{i} B$ on types
$A,B : {[]}_{i} $ of a common universe level $i$, 
and a heterogeneous relation $a \precision{A}{B} b$ between terms $a : A$ and
$b : B$.
These two precision relations are internalized as two new primitive type formers,
%
and their content is described by their behaviour on
their type parameters.
In practice we present these relations through a confluent reduction system
in~\cref{fig:ccic-prec}, corresponding to
a definition by case analysis on the type parameters,
which is how the model of~\cref{sec:model} proceeds.
We note $a \equiprecision{A} a'$ for $a \precisionDiag{A} a' \wedge a'
\precisionDiag{A} a$.

%
Let us now explain the
two main properties we expect to hold.
%
%
First, the precision relation should be \emph{transitive}:
there should be an operation $\prectrans$ such that if $e : A
\precisionType{i} B$ and $e' : B \precisionType{i} C$ then $e \prectrans e' : A
\precisionType{i} C$.
Second, the precision relation cannot be reflexive. Indeed, reflexivity
at function types $A \to B$ entails monotonicity: due to the way we 
define precision, if a function $f : A
\to B$ verifies $f \precision{}{} f$ then for any $a \precision{A}{A} a'$, $f\,a
\precision{B}{B} f\,a'$. But we do not want to globally forbid such
non-monotone features, as we rather made the design choice to allow some
non-monotonicity in \CCICPrec, \textit{e.g.} the $\catch$ construct.
%
As a consequence, reflexivity becomes a property, and
we say that a type $A : {[]}_{i}$ is \emph{self-precise},
noted $\spType{A}{i}$, when it is a reflexive element of $\precisionType{i}$.
Similarly, a term $a : A$ is called
\emph{self-precise}, noted $\sp{a}{A}$, when it is related to itself by
$\precision{A}{A}$.
Not every type is self-precise, but the precision relation is
\emph{quasi-reflexive}: if two types $A,B$ are related by precision $e : A
\precisionType{i} B$,
both are self-precise,\footnote{In order to obtain
  transitivity on function types, the precision relation needs to be
  at least co-transitive, a property obtained here as a consequence of
  quasi-reflexivity.} so we have self-precision proofs $\lrefl{e} :
\spType{A}{i}, \urefl{e} : \spType{B}{i}$.

\begin{figure}
  \begin{small}
  \begin{mathpar}
    \inferrule[$\sqsubseteq$-Type-Wf]
    {\Gamma \vdash A, B : []_i}
    {\Gamma{} \vdash A \precisionType{i} B : \prop}
    \ilabel{infrule:prec-type-wf}
    \and
    \inferrule[$\sqsubseteq$-Wf]
    {\Gamma{} \vdash A,B : []_i\\
      \Gamma{} \vdash t : A \\
      \Gamma{} \vdash u : B}
    {\Gamma{} \vdash t \precision{A}{B} u : \prop}
    \ilabel{infrule:prec-wf}
    \\
    \inferrule[${[]}$-Refl-Ty]
    {\Gamma \vdash}
    {\Gamma \vdash \congTm{[]_i} : \selfprecisionType{[]_i}{i+1}}
    \ilabel{infrule:univ-prec-refl-ty}
    \and
    \inferrule[$\prop$-Refl-Ty]
    {\Gamma \vdash}
    {\Gamma \vdash \congTm{\prop} : \selfprecisionType{\prop}{0}}
    \ilabel{infrule:prop-prec-refl-ty}
    \and
    \inferrule[$\cum$-Cong-Ty]
    {}
    {\cum A \precisionType{i+1} \cum B \redCCIC A \precisionType{i} B}
    \ilabel{infrule:cum-prec-cong-ty}
    \and
    \inferrule[$\listT$-Cong-Ty]
    {}
    {\listT A \precisionType{i} \listT B \redCCIC A \precisionType{i} B}
    \ilabel{infrule:list-prec-cong-ty}
    \and
    \inferrule[$\P$-Cong]
    {}
    {\P x : A. B \precisionType{i} \P x : A'. B' \redCCIC A \precisionType{i} A' \wedge
      {\left \{
          {\arraycolsep=1.5pt
            \begin{array}{lrclll}
             \forall a_0\, a_1,& a_0 \precision{A}{A} a_1 & \to & B[a_0/x] &\precisionType{i} B[a_1/x] &\wedge\\
             \forall a'_0\, a'_1,& a'_0 \precision{A'}{A'} a'_1 & \to & B'[a'_0/x'] &\precisionType{i} B'[a'_1/x'] &\wedge\\
             \forall a\, a',& a \precision{A}{A'} a' & \to & B[a/x] &\precisionType{i} B'[a'/x']&
           \end{array}}
        \right .}
    }
    \ilabel{infrule:pi-prec-cong}
    \\
    \inferrule[${[]}$-$\sqsubseteq$]
    {}
    {A \precision{{[]_{i}}}{{[]_{i}}} B\redCCIC A \precisionType{i} B \wedge B \precisionType{i} \?_{[]_i}}
    \ilabel{infrule:univ-prec-def}
    \and
    \inferrule[${[]}$-$\?$-Bound]
    { \Gamma \vdash }
    {\Gamma \vdash \unkUpTm{[]_i} : []_i \precisionType{i+1} \?_{[]_{i+1}}}
    \ilabel{infrule:univ-prec-unk}
    \and
    \and
    \inferrule[$\prop$-$\?$-Bound]
    { \Gamma \vdash }
    { \Gamma \vdash \unkUpTm{\prop} : \prop \precisionType{0} \?_{[]_0}}
    \ilabel{infrule:prop-prec-unk}
    \and 
    \inferrule[$\cum$-$\?$-Bound]
    {\Gamma{} \vdash w : \selfprecisionType{A}{i}}
    {\Gamma{} \vdash \unkUpTm{\cum}\,w : \cum\,A \precisionType{i+1} \?_{[]_{i+1}}}
    \ilabel{infrule:cum-prec-unk}
    \and
    \inferrule[$\listT$-$\?$-Bound]
    {\Gamma{} \vdash w : A \precisionType{i} \?_{[]_i}}
    {\Gamma{} \vdash \unkUpTm{\listT}\,w : \listT\,A \precisionType{i} \?_{[]_i}}
    \ilabel{infrule:list-prec-unk}
    \and
    \inferrule[$\err$-Refl]
    { \Gamma{} \vdash w : \spType{A}{i}}
    {\Gamma{}\vdash \reflPrec{\err{}}\,w : \selfprecision{\err_{A}}{A}}
    \ilabel{infrule:err-prec-refl}
    \and
    \inferrule[$\?$-Refl]
    { \Gamma{} \vdash w : \spType{A}{i}}
    {\vdash \reflPrec{\?}\,w : \selfprecision{\?_{A}}{A}}
    \ilabel{infrule:unk-prec-refl}
    \and
    \inferrule[$\err$-$\sqsubseteq$]
    { \Gamma{} \vdash w_A : \spType{A}{i}\\
      \Gamma{} \vdash w_B : \spType{B}{i}\\
      \Gamma{} \vdash w_b : \sp{b}{B}}
    {\Gamma{} \vdash \errMin\,w_A\,w_B\,w_b : \err_{A} \precision{A}{B} b}
    \ilabel{infrule:prec-err}
    \and
    \inferrule[$\?$-$\sqsubseteq$]
    { \Gamma{} \vdash w_A : \spType{A}{i}\\
      \Gamma \vdash w_a : \sp{a}{A} \\ 
      \Gamma{} \vdash w_B : \spType{B}{i}}
    {\Gamma{} \vdash \unkMax\,w_A\,w_a\,w_B : a \precision{A}{B} \?_{B}}
    \ilabel{infrule:prec-unk}
    \\
    \inferrule[$\prop$-$\sqsubseteq$]
    { \Gamma{} \vdash P : \prop\\ \Gamma{} \vdash Q : \prop}
    {\Gamma{} \vdash \propIrr\,P\,Q : P \precisionDiag{\prop} Q}
    \ilabel{infrule:prop-prec-def}
    \and
    \inferrule[$\BoxProp$-Cong]
    {}
    {\BoxProp\,P \precisionType{0} \BoxProp\,Q \redCCIC P \precision{\prop}{\prop} Q}
    \ilabel{infrule:box-prec-cong}
    \and
    \inferrule[$\BoxProp$-$\sqsubseteq$]
    {\Gamma \vdash b : \BoxProp\,P\\ \Gamma{} \vdash b' : \BoxProp\,Q}
    {\Gamma \vdash \boxPrecIrr\,b\,b' : b \precision{\BoxProp\,P}{\BoxProp\,Q} b'}
    \ilabel{infrule:box-prec} 
    \\
    \inferrule[$\listT$-$\sqsubseteq$-$\nilK$]
    { \Gamma{} \vdash A : [] }
    {\Gamma \vdash \congTm{\nilK}\,A : \selfprecision{\nilK}{\listT{A}}}
    \ilabel{infrule:list-prec-nil}
    \and
    \inferrule[$\listT$-$\sqsubseteq$-$\consK$]{}{\consK\,a\,l \precision{\listT\,A}{\listT\,A'}
      \consK\,a'\,l \redCCIC a \precision{A}{A'} a' \wedge l \precision{\listT\,A}{\listT\,A'} l'}
    \ilabel{infrule:list-prec-cons}
    \\
    \inferrule[NoConf-$\nilK$-$\consK$]
    {}
    {\nilK \precision{\listT\,A}{\listT\,A'} \consK\,a\,l \redCCIC \botProp}
    \ilabel{infrule:list-prec-noconf-nil-cons}
    \and
    \inferrule[NoConf-$\consK$-$\nilK$]
    {}
    {\consK\,a\,l \precision{\listT\,A}{\listT\,A'} \nilK \redCCIC \botProp}
    \ilabel{infrule:list-prec-noconf-cons-nil}
    \and
    \inferrule[$\Pi$-$\sqsubseteq$]
    {}
    {f \precision{\P x : A. B}{\P x : A'. B'} g \redCCIC
      {\left \{
          {\arraycolsep=1.5pt
            \begin{array}{lrcll}
              \forall a_0\, a_1,& a_0 \precision{A}{A} a_1 & \to & f\,a_0 \precision{B\,a_0}{B\,a_1} f\,a_1 &\wedge\\
              \forall a'_0\, a'_1,& a'_0 \precision{A'}{A'} a'_1 & \to & g\,a'_0 \precision{B'\,a'_0}{B'\,a'_1} g\,a'_1 &\wedge\\
              \forall a\, a',& a \precision{A}{A'} a' & \to & f\,a \precision{B\,a}{B'\,a'} g\,a'&
            \end{array}}
        \right .}}
    \ilabel{infrule:pi-prec-def} 
    \and
    \inferrule[$\cum$-$\sqsubseteq$]
    {}
    {a \precision{\iota A}{\iota B} b \redCCIC \coeinv a \precision{A}{B}
      \coeinv b}
    \ilabel{infrule:coe-prec-cong}
    
  \end{mathpar}
  \end{small}
  \caption{Precision on types and terms}
  \label{fig:ccic-prec}
\end{figure}

Let us now turn to the actual content of the precision relations as defined
in~\cref{fig:ccic-prec}. \shnew{Term and type precision are internally supported by adding two 
new term formers, whose 
typing is given by} the first two rules \nameref{infrule:prec-type-wf} and
\nameref{infrule:prec-wf}. 
%
\nameref{infrule:univ-prec-refl-ty} next states that each universe ${[]}_{i}$ is
self-precise (as a type), \nameref{infrule:prop-prec-refl-ty} that
$\prop$ is self-precise at level $0$, and \nameref{infrule:cum-prec-cong-ty} and
\nameref{infrule:list-prec-cong-ty} that $\cum$ and $\listT$ are congruent for
precision on types at the adequate levels.
Precision at product types is the crux of the definition of precision, we
defer its explanation of \nameref{infrule:pi-prec-cong} to after the other rules.
For now, it is only important to note that contrarily to other type formers,
there is no rule to relate product types \emph{as terms}, only \emph{as types}.
This is the technical counterpart of the intuition given in
\cref{sec:internal-prec} that precision
between products should be guarded by an explicit use of cumulativity.
Next come the rules for type formers as terms: all of them---apart, crucially,
from product types--- are either directly self-precise (as terms of $[]_i$) or
congruent because they are congruent for type precision and bounded above by
$\?_{[]}$.
Indeed, heterogeneous precision between types reduces to homogeneous
precision between types more precise than $\?_{[]_i}$ by virtue of
\nameref{infrule:univ-prec-def}, tying the knot between the two notions.
As a consequence, a proof of precision $A \precisionDiag{[]_i} A$ entails
that $A \precisionType{i} A$ as well as $A \precisionType{i} \?_{[]_i}$.
$[]_i$, $\prop$ are bounded by $\?_{[]}$ via \nameref{infrule:univ-prec-unk},
\nameref{infrule:prop-prec-unk}, whereas
$\cum$ require that its parameter is self-precise, rule \nameref{infrule:cum-prec-unk}, and $\listT$ that its
parameter is bounded by $\?_{\square}$, rule \nameref{infrule:list-prec-unk}.
The two exceptional types $\err_{[]}$ and $\?_{[]}$ are also self-precise, both
as types and terms of the universe, using instances of \nameref{infrule:err-prec-refl} and
\nameref{infrule:unk-prec-refl}.
%
%

%
More generally, the rules \nameref{infrule:err-prec-refl} and
\nameref{infrule:unk-prec-refl} ensure that the terms $\err_{A}$ and $\?_{A}$
are in relation with themselves, while \nameref{infrule:prec-err} and
\nameref{infrule:prec-unk} say that they are respectively minimal and
maximal---for self-precise terms of a self-precise type.
Heterogeneous precision between
propositions is degenerate 
(\nameref{infrule:prop-prec-def}), meaning that any two propositions are related
by precision.
Monotonicity of $\BoxProp$ with respect to precision on propositions
(\nameref{infrule:box-prec-cong}) means that precision between boxed propositions
is degenerate as well.
To validate this, we endow $\BoxProp$
types with a precision relation collapsing all terms (\nameref{infrule:box-prec}).
This is sensible, as it
showcases the fact that no (self-precise) context should be allowed to
distinguish two proofs of a proposition, since those, even $\BoxProp$ed, ought to be
observationally subsingletons. It also makes the eager erroring behavior of
\nameref{redrule:box-box} sensible, since the error is as good an inhabitant
of a $\BoxProp$ed proposition as any.

Cumulativity preserves the relation between types coming
from lower levels (\nameref{infrule:coe-prec-cong}), meaning that coercions
between a type and its lifting are monotone.
On inductive types the precision relation closely resembles binary
parametricity~\cite{bernardyAl:jfp2012}, relating a constructor to itself when
arguments are related (\nameref{infrule:list-prec-nil},
\nameref{infrule:list-prec-cons}).
Two no confusion principles
(\nameref{infrule:list-prec-noconf-nil-cons},
\nameref{infrule:list-prec-noconf-cons-nil}) allow to deny the relatedness of
lists that have distinct head constructors.\footnote{ In the case of lists and
  using transitivity, we can derive solely from these two rules that any
  non-exceptional constructor is discriminable from $\err_{\listT\,A},
  \?_{\listT\,A}$, \textit{e.g.} that $\nilK \notprecisionDiag{\listT\,A}
  \err_{\listT\,A}$, and $\?_{\listT\,A} \notprecisionDiag{\listT\,A}
  \err_{\listT\,A}$.
For other inductive types such as $\zeroType$ or $\unit$, these rules should be assumed primitively,
\textit{e.g.} $\?_{\zeroType} \notprecisionDiag{\zeroType} \err_{\zeroType}$
for the empty type $\zeroType$.}

Finally, we need to explain how function types are related by (type) precision.
For simplicity, we start with the non-dependent case that takes the standard
shape found in other gradual languages: two function types $A \to B$ and $A' \to
B'$ are related whenever their domains and codomains are related: $A \precisionType{i} A' \wedge B \precisionType{i} B'$.
%
The relation of precision $f \precision{A{\to}B}{A'{\to}B'} g$ between functions
$f : A \to B$ and $g : A' \to B'$ has to ensure that (1) $f$ is monotone with
respect to the precision on $A$ and $B$; (2) $g$ is monotone with respect to the
precision on $A'$ and $B'$; and (3) given inputs $a:A, a' : A'$ related by
precision $a \precision{A}{A} a'$, $f\,a : B $ is related to $g\,a' : B'$ by
$\precision{B}{B}$.
%
Condition (3) boils down to the standard definition of (binary) parametricity on
function types.
Additional conditions (1-2) are required to ensure quasi-reflexivity at
function types: since we do not want to globally impose that functions respect
precision, we need to explicitly require that precision only relates monotone
functions.
For a function $f : A \to B$ between self-precise types, being
self-precise is logically equivalent to being monotone with respect to
precision, so conditions $(1$-$3)$ are equivalent in that case.
In the case of dependent function types (\nameref{infrule:pi-prec-cong}),
domains must be related similarly to the non-dependent case but the codomains
must now be related \emph{as type families}, meaning that they are required to
satisfy variants of the conditions (1-3) with respect to type precision.
Finally, the relation between dependent functions is described by 
\nameref{infrule:pi-prec-def} and requires again that both functions are monotone
and map related input to related outputs, at the adequate types.

\begin{example}[Necessity of monotonicity in function types]
  \label{ex:monotonicity-function-types}
  Consider the two functions of type $\zeroType \to \unit$  given by
    $f := \catchat{\zeroType}~(\lambda (x :
    \zeroType). \unit)~\unitK~\err_{\unit}$
    and
    $g := \catchat{\zeroType}~(\lambda (x : \zeroType).
    \unit)~\?_{\unit}~\unitK$ using  the eliminator for the empty inductive type
    $\zeroType$, $\catchat{\zeroType} : \P(P : \zeroType \to
    \univ{})(h_{\err} : P\,\err_{\zeroType})(h_{\?} : P\,\?_{\zeroType})(x : \zeroType). P\,x$.
  %
  These functions verify that $\forall x \precisionDiag{\zeroType} y, f\,x
  \precisionDiag{\unit} g\,y$, but neither $f$ or $g$ are monotone.
  As a consequence, precision on function types need to be restricted to
  monotone functions.
  Taking $f$ to be instead the constant function with value $\err_{\unit}$, or $g$
  the constant function with value $\?_{\unit}$ shows that we really need both
  functions to be monotone.
\end{example}

\subsection{Properties of Precision}
\label{sec:precision-properties}

We now extend the theory with properties about precision that are validated by our model (presented in \cref{sec:model}), in order to allow users to reason abstractly about
precision proofs in \CCICPrec. Thus, whenever we say that a property ``holds''
in this section, it should be understood as a twofold statement:
first, the property is validated in the model, and so
we add a new constant in \CCICPrec, witnessing its truth.

\paragraph{Embedding-projection pairs}
Why do we care so much about precision?
The fundamental reason is that casts between types that are related by precision
are well-behaved.
We adopt the approach of~\citet{newAhmed:icfp2018} to characterize well-behaved
pairs of casts as those that form an embedding projection pair (\eppair).
In our setting that allows non monotone functions, the definition of an
\eppair needs to be relativized to self-precise elements.
\begin{definition}[Embedding projection pairs]
  \label{def:eppair}
  A pair of functions $(\ascdom{}{A}{B} : A \to B, \ascdom{}{B}{A} : B \to A)$
  is an embedding projection pair, notation $\ascdom{}{A}{B} \dashv
  \ascdom{}{B}{A}$
, when:
  \begin{description}
  \item[Monotonicity] both $\ascdom{}{A}{B}$ and $\ascdom{}{B}{A}$ are monotone
    with respect to precision,
    \begin{align*}
      \forall a\,a' : A, a \precisionDiag{A} a' \to \ascdom{a}{A}{B} \precisionDiag{B} \ascdom{a'}{A}{B} \\
      \forall b\,b' : B, b \precisionDiag{B} b' \to \ascdom{b}{B}{A} \precisionDiag{B} \ascdom{b'}{B}{A} 
    \end{align*}
  \item[Adjunction] for any self-precise terms $a : A, b : B$ the following adjunction
    property is verified
    \[\sp{a}{A} \wedge \sp{b}{B} \quad\to\quad \ascdom{a}{A}{B}
      \precisionDiag{B} b \leftrightarrow a \precisionDiag{A} \ascdom{b}{B}{A},\]
  \item[Retraction]  a self-precise term $a : A$ is equiprecise
  with its downcast-upcast:
    \[\sp{a}{A} \quad \to \quad \ascdom{\ascdom{a}{A}{B}}{B}{A} \precisionDiag{A} a \]
    \shnew{The reverse precision relation 
is a consequence of reflexivity and the
  adjunction property.}
  \end{description}
  We call $\ascdom{-}{A}{B} : A \to B$ the \textbf{upcast}
  associated to the \eppair and $\ascdom{-}{B}{A} : B \to A$ the \textbf{downcast}.
\end{definition}
\shepherd{}{
  \begin{proposition}
    \label{prop:grip-graduality}
    In \CCICPrec, any pair of casts $(\ascdom{}{A}{B} : A \to B, \ascdom{}{B}{A} : B
    \to A)$ between types  $A \precisionType{i} B$ related by precision forms an
    embedding projection pair witnessed by
    \[\preceppair \quad:\quad \forall A\, B,\quad A \precisionType{i} B \quad\to\quad \ascdom{}{A}{B} \dashv \ascdom{}{B}{A}.\]
  \end{proposition}
  \begin{proof}
    The addition of the constant $\preceppair$ is justified by the model of \CCICPrec presented
    in \cref{sec:model}, in particular by the functorial component of $\El$ in
    \cref{thm:properties-precision} providing an \eppair for any two types related
    by precision.
  \end{proof}
}
\begin{figure}
  \flushleft
\begin{small}
  \emph{Quasi-reflexivity and transitivity} 
  \begin{align*}
    &&\text{Implicit} &\text{ bindings }
      w_A : \selfprecisionType{A}{i}, w_B : \selfprecisionType{B}{i}, w_C : \selfprecisionType{C}{i}:\\
    \lrefl{-} &:A \precisionType{i} B \to A \precisionType{i} A
      & \lrefl{-} &: \forall \{w_A\,w_B\}, a \precision{A}{B} b \to a \precisionDiag{A} a\\
    \urefl{-} &: A \precisionType{i} B \to B \precisionType{i} B
      & \urefl{-} &: \forall \{w_A\,w_B\}, a \precision{A}{B} b \to b \precisionDiag{B} b\\
    -\prectrans-~ &: A \precisionType{i} B \to B \precisionType{i} C \to A \precisionType{i} C
      & - \prectrans-~ &:\forall \{w_A\,w_B\,w_C\}, a \precision{A}{B} b \to b \precision{B}{C} c \to a \precision{A}{C} c
  \end{align*}
  \medskip
  
  \emph{Decomposition of casts}
  \begin{mathpar}
    \inferrule[Upper-decomposition]
    { \Gamma \vdash w_{AX} : A \precisionType{i} X\\
      \Gamma \vdash w_{BX} : B \precisionType{i} X\\
      \Gamma \vdash w_a : \sp{a}{A}}
    { \Gamma \vdash \upperdecompl\,w_{AX}\,w_{BX}\,w_a : \ascdom{\ascdom{a}{A}{X}}{X}{B} \equiprecision{B} \ascdom{a}{A}{B}}
    \ilabel{infrule:cast-decomposition}
  \end{mathpar}
  \medskip

  \emph{Decomposition of heterogenous term precision}
  \begin{align}
    \text{For}\,A \precisionType{i} X, B \precisionType{i} X, \qquad a \precision{A}{B} b \qquad \leftrightarrow \qquad \selfprecision{a}{A} \quad\wedge\quad\ascdom{a}{A}{X} \precisionDiag{X} \ascdom{b}{B}{X} \quad  \wedge \quad \selfprecision{b}{B}
    \label{eq:heterogeneous-precision-decomposition}
  \end{align}
\end{small}
  \caption{Axioms of precision}
  \label{fig:ccic-prec-axioms}
\end{figure}
\begin{figure}
  \flushleft
  \begin{small}
  \emph{Functoriality \& monotonicity of casts.}
  \begin{mathpar}
    \inferrule[Cast-Id]
    {\selfprecisionType{A}{i}\\\selfprecision{a}{A}}
    {\ascdom{a}{A}{A} \equiprecision{A} a} 
    \ilabel{infrule:cast-id}
    \and
    \inferrule[Upcast-Comp]{
      A \precisionType{i} B\\
      B \precisionType{i} C\\
      \selfprecision{a}{A}}
    {\ascdom{\ascdom{a}{A}{B}}{B}{C} \equiprecision{C} \ascdom{a}{A}{C}} 
    \ilabel{infrule:upcast-comp}
    \and
    \inferrule[Downcast-Comp]{
      A \precisionType{i} B\\
      B \precisionType{i} C\\
      \selfprecision{c}{C}}
    {\ascdom{\ascdom{c}{C}{B}}{B}{A} \equiprecision{A} \ascdom{c}{C}{A}} 
    \ilabel{infrule:downcast-comp}
    \and
    \inferrule[Cast-Mon]{}{\sp{\ascdom{}{-}{-}}{\P(A\,B :
        \univ{i}).A{\to}B}\equiv {%
        \begin{array}{l}
          \forall A\,A'\,(w_A : A \precisionType{i} A')\,B\,B'\,(w_B : B
          \precisionType{i} B')\\
          ~(a : A)(a' : A') (w_a : a \precision{A}{A'} a'),\ascdom{a}{A}{B} \precision{B}{B'} \ascdom{a'}{A'}{B'}         
        \end{array}}}
  \end{mathpar}
  \medskip

  \emph{Characterization of heterogenous term precision}
  \begin{align}
    \text{For}\,\selfprecisionType{A}{i}, \selfprecisionType{B}{i}, \qquad a \precision{A}{B} b \qquad \leftrightarrow \qquad a \precisionDiag{A} \ascdom{b}{B}{A} \quad  \wedge \quad \selfprecision{b}{B}
    \label{eq:heterogeneous-precision}
  \end{align}
  \end{small}
  \caption{Properties of precision}
  \label{fig:ccic-prec-properties}
\end{figure}

\paragraph{Order-like properties} In order to establish that two types are
related by precision, we can use the generic axioms of the precision relations
described in \cref{fig:ccic-prec-axioms} beside those of \cref{fig:ccic-prec}.
Type precision is a quasi-reflexive and transitive relation,
and so is term precision at any self-precise type,
meaning that $\precisionDiag{A}$ is quasi-reflexive and transitive whenever
$\spType{A}{i}$. Moreover, using \cref{fig:ccic-prec}, they admit $\err$ and
$\?$ as respectively smallest and largest (self-precise) elements.
More generally, heterogeneous term precision satisfies indexed variants of
quasi-reflexivity and transitivity on self-precise types.

\paragraph{Decomposition of casts and heterogeneous precision}
A further fundamental property of casts is that they decompose through any type
less precise than both the source and the target of the cast: if $A \precisionType{i} X$
and $B \precisionType{i} X$, then for any self-precise term $a : A$, the cast
$\ascdom{a}{A}{B}$ is equiprecise to an upcast from $A$ to $X$ followed by a downcast to $B$:
\[\sp{a}{A} \to \ascdom{\ascdom{a}{A}{X}}{X}{B}
  \equiprecision{B} \ascdom{a}{A}{B}\]
Heterogenous term precision $\precision{A}{B}$ satisfy a similar decomposition
property~\cref{eq:heterogeneous-precision-decomposition} expressing that the relation between self-precise elements can be
reduced to homogeneous precision at any common upper bound $X$ of $A,B$ for type
precision.
%
%
In particular, whenever $A, B : []_i$ are more precise than $\?_{[]_i}$,
that is when $A,B$ are self-precise as terms of $[]_i$, $\?_{[]_i}$
provides such a common upper bound for precision.
As long as precision and cast are concerned, self precise types $A, B : []_i$
that are not bounded by $\?_{[]_i}$ can be adequately replaced by $\cum\,A$ and
$\cum\,B$, thanks to \nameref{infrule:cum-prec-cong-ty} and \nameref{infrule:coe-prec-cong}, for which $\?_{[]_{i+1}}$ is an upper bound.
As a consequence of these properties, heterogeneous term precision between self
precise types can be reformulated using solely homogeneous precision at $A,B$
and casts:
\begin{align*}
  a \precision{A}{B} b \qquad&\leftrightarrow\qquad  \selfprecision{a}{A} \>\wedge\> \ascdom{\coe\,a}{\cum A}{\?_{[]}} \precisionDiag{\?_{[]}} \ascdom{\coe\,b}{\cum B}{\?_{[]}} \>\wedge\> \selfprecision{b}{B}\\
  &\leftrightarrow\qquad a \precisionDiag{A} \ascdom{b}{B}{A}  \>\wedge\> \selfprecision{b}{B}                                                                                                                                            
\end{align*}

\paragraph{Composing casts}
Using \nameref{infrule:cast-decomposition} and the monotonicity of
embedding projection pairs, we can show that the \eppair induced
by precision are functorial: casting a self-precise term $a$ of a self-precise
type $A$ to $A$ itself is equiprecise to $a$ (\nameref{infrule:cast-id}), a
succession of upcasts between precision-related types combine to a single upcast
(\nameref{infrule:upcast-comp}) and similarly for downcasts
(\nameref{infrule:downcast-comp}).
%
%

\paragraph{Failure of threesomes}
Since casts decompose in a well-behaved way through any upper bound,
it is natural to wonder whether a
similar property would hold for lower bounds, as can be found in
threesomes~\cite{siekAl:popl10} in the simply-typed gradual setting.
In general, if $Y \precisionType{i} A$, $Y \precisionType{i} B$ we can
derive from properties of casts that for any self-precise term $a : A$,
$\ascdom{\ascdom{a}{A}{Y}}{Y}{B} \precisionDiag{B} \ascdom{a}{A}{B}$, and taking
$A = B = \nat$, $Y = \err_{[]_i}$ and $a = 0$ shows that this precision ordering
can be strict.
We could still expect that this relation is an equiprecision when $Y$ is
sufficiently close to both $A$ and $B$, typically when it is their meet $A \sqcap B$ for the precision relation.
Such a condition is known as the Beck-Chevalley condition in the literature on
hyperdoctrines and descent~\cite{Law70}, and the following counterexample shows
that this property does not hold in \CCICPrec.

\begin{example}[No cast decomposition through meets]
  \ilabel{casts-meet-decomposition}
  Computing the meet of $X_1 = \nat \to \nat$ and
  $X_2 = \Pi(b : \mathbb{B})(\ifte{b}{\nat}{\bool})$ gives
  \begin{align*}
    X_1 \sqcap X_2 &= \Pi(x : \nat \sqcap \bool)~\nat \sqcap (\ifte{\ascdom{x}{\nat \sqcap \bool}{\bool}}{\nat}{\bool})\\
                   &= \Pi(x : \err_{[]}) ~\nat \sqcap (\ifte{\err_{\bool}}{\nat}{\bool})\\
                   &= \Pi(x : \err_{[]}) ~\nat \sqcap \err_{[]} \\ & = \err_{[]} \to \err_{[]}
  \end{align*}
  Now computing the result of casting $f : X_{1} := \lambda(n : \nat). 5$ to $X_{2}$ directly and through $X_{1} \sqcap X_{2}$,
  and evaluating both results on $\btrue$, we obtain
  \begin{align*}
    (\ascdom{f}{X_{1}}{X_{2}})~\btrue
    &= (\lambda(b : \bool). \ascdom{f\,(\ascdom{b}{\bool}{\nat})}{\nat}{\ifte{b}{\nat}{\bool}})~\btrue \\
    &= (\lambda(b : \bool). \ascdom{f\,\err_{\nat}}{\nat}{\ifte{b}{\nat}{\bool}})~\btrue \\
    &= (\lambda(b : \bool). \ascdom{5}{\nat}{\ifte{b}{\nat}{\bool}})~\btrue \\
    &= \ascdom{5}{\nat}{\ifte{\btrue}{\nat}{\bool}}= \ascdom{5}{\nat}{\nat} = 5
  \end{align*}
  and
  \begin{align*}
    (\castmid{X_{2}}{X_{1} \sqcap X_{2}}{X_{1}}{f})\,\btrue
    &= (\castmid{X_{2}}{\err_{[]}{\to}\err_{[]}}{X_{1}}{f})\,\btrue\\
    &= (\cast{\err_{[]}{\to}\err_{[]}}{X_{2}}{\lambda(x : \err_{[]}). \err_{\err_{[]}}})~\btrue\\
    &= (\lambda(b : \bool). \err_{\ifte{b}{\nat}{\bool}})~\btrue = \err_{\nat}
  \end{align*}
  \shnew{Note that for these examples the call-by-name behavior of $\err$ \cite{pedrotTabareau:esop2018} is crucial.}
  In particular, $\cast{X_{1}}{X_{2}}{f} \not\sqsubseteq \castmid{X_{1}}{X_{1} \sqcap X_{2}}{X_{2}}{f}$ and the cast from $X_{1}$ to $X_{2}$ cannot be decomposed through a type more precise than both $X_{1}$ and $X_{2}$.
  This counterexample can be adapted to use dependent sums $\Sigma$ instead of
  dependent products, showing that this phenomenon is proper to type dependency
  and function types are not crucial.
\end{example}

Note that all the properties presented in this section only apply to
self-precise terms.
The behavior of cast on types or terms that are not self-precise, typically non
monotone functions, is left partially unconstrained.
%

\shnew{
  \paragraph{Dynamic Gradual Guarantee}

  A crucial property of precision is that self-precise contexts
  (\ie functions for a type $A$ to $\bool$) are monotone. As explained in \cref{sec:action-dgg},
  this is a form of Dynamic Gradual Guarantee, 
  and it follows directly from the definition of precision for functions.

  \begin{theorem}[Dynamic Gradual Guarantee]
    \label{thm:dgg}
    For any $A : []$ and boolean context $C : A \to \bool$ such that $\sp{C}{A \to \bool}$,
    if $x, y : A$ are such that $x \precision{A}{A} y$, it also holds that $C~x \precision{\bool}{\bool} C~y$.
  \end{theorem}
}

\subsection{Monotone Fragment}
\label{sec:prec-monotone-frag}
By adequately restricting \CCICPrec, we can consider a fragment where every term
is monotone.
On that fragment, precision between functions only needs a single
heterogeneous component, bypassing boilerplate proofs of monotonicity.
In practice, a characterization of this fragment could be used to automatically
synthesize monotonicity proofs and lift a sizeable share of the
burden imposed to the programmer.
\shepherd{
The quest for a large monotone fragment explains our design decision to
degenerate precision on $\prop$: doing so allows for a monotone negation on
propositions.
}{}
%

There are two main non-monotone features in \CCICPrec.
\shepherd{}{
The $\catch$ constructor, which purposely allows for a non-monotone treatment of
$\err$ and $\?$ (see \cref{ex:monotonicity-function-types}), is the first source
of non-monotone terms.
The second source of non-monotone terms lie in the use of $\Pi$ to produce terms
of a universe, which cannot be monotone due to the Fire Triangle of Graduality.
However, \citet{lennonAl:toplas2022} explain how to sidestep the latter
obstruction by systematically lifting $\Pi$ types by one universe level up, a
soluution employed} in
their \CCICs system---the only variant of \CCIC that satisfies both
normalization and graduality, by sacrificing conservativity over \CIC\shepherd{,
  ---which we can embed in \CCICPrec}{}. We can rethink \CCICs as an attempt to
guarantee that every well-typed term is self-precise in order to globally
satisfy graduality. \shnew{Inspired by this technique, we construct \CCICPrecs,
  a subsystem of \CCICPrec where every term is self-precise.}

\paragraph{Monotone $\catch$}
The typical non-monotone construction in \CCICPrec, is the $\catch$ construction
on inductive types (see \cref{ex:monotonicity-function-types}).
However there is a generic way to prove that a $\catch$ is monotone,
assuming adequate precision hypotheses on its arguments.
In the case of lists, monotonicity of $\lcatch{A}$ amounts to:
\[\forall l~l',~ l\precisionDiag{\listT\,A} l' {\to}
  \lcatch{A}\,P\,h_{\nilK}\,h_{\consK}\,h_{\err}\,h_\?\,l
  \precision{P\,l}{P\,l'} 
  \lcatch{A}\,P\,h_{\nilK}\,h_{\consK}\,h_{\err}\,h_\?\,l'
 \]
A natural proof \shepherd{}{of monotonicity} proceeds by successive induction on $l$ and $l'$ using
$\catch^{\prop}$.
The cases with distinct head constructors, e.g. $l = \nilK, l' =
\consK\,a\,l''$, are contradictory thanks to the no-confusion rules for
precision on list (for instance \nameref{infrule:list-prec-noconf-nil-cons}).
For the valid cases, we need to assume that the the branches $h_{\nilK}$ and
$h_{\consK}$ are less precise than $h_{\err}$ and more precise than $h_{\?}$,
and that $h_{\consK}$ is self-precise, e.g. $h_{\err}
\precision{P\,\err_{\listT\,A}}{P\,\nilK} h_{\nilK}$.
\shnew{In particular, 
$\lind\,P\,h_{\nilK}\,h_{\consK} :=
  \lcatch{A}\,P\,h_{\nilK}\,h_{\consK}\,\err_{P\,\err_{\listT\,A}}\,\?_{P\,\?_{\listT\,A}}$
is always monotone if $P$, $h_{\nilK}$ and $h_{\consK}$ are self-precise.}


%
\shepherd{
\paragraph{\CCICs, a gradual fragment of \CCICPrec}
We can achieve an embedding of \CCICs into \CCICPrec
by using explicit cumulativity to systematically increase
the universe level of $\Pi$ types, mimicking the typing rule of products in
the former system:}
{
\paragraph{\CCICPrecs, a gradual fragment of \CCICPrec}
In \citet{lennonAl:toplas2022}, the system \CCICs is both gradual and normalizing,
at the cost of being more conservative than \CIC: some terms are typable in \CIC, but
not in \CCICs. This is done by systematically increasing the level of a $\Pi$ type.
Drawing inspiration from this, we can define \CCICPrecs, which has exactly the same
rules for typing and conversion as \cref{fig:ccic-base,fig:ccic-reduction},
but for rule \nameref{infrule:cic-prod} replaced by the following
rule \nameref{infrule:ccics-prod}, and uses of $\lcatch{A}$ restricted to $\lind{A}$ as defined above.}
\[
  \inferrule[$\P$-\CCICPrecs]
  {\Gamma \ccicprecsty A : []_{i} \\
    \Gamma, x : A \ccicprecsty B : []_{i}}
  {\Gamma \ccicprecsty  \P x : A. B : []_{i+1}}
  \ilabel{infrule:ccics-prod}
\]
To distinguish the two, we use $\ccicprecty$ for judgments in \CCICPrec,
and $\ccicprecsty$ for judgments in \CCICPrecs\shepherd{, that
consist of the rules from \cref{fig:ccic-base,fig:ccic-reduction},
with the exception of rule \nameref{infrule:cic-prod} replaced by
rule \nameref{infrule:ccics-prod}
above}{}.
%
It is rather straightforward to define a translation $\shiftTl{-}$ from \CCICPrecs
to \CCICPrec: the translation preserves all term and type constructor but $\Pi$
types where it adds an explicit coercion \shepherd{of}{due to} cumulativity:
\[
  \begin{array}{lcl}
    \shiftTl{\P x : A. B} &:=& \cum\,(\P x : \shiftTl{A}. \shiftTl{B}) \\
    \shiftTl{\l x : A. t} &:=& \coe\,(\l x : \shiftTl{A}. \shiftTl{t}) \\
    \shiftTl{t~u} &:=& (\coeinv\,\shiftTl{t})~\shiftTl{u}
  \end{array}
\]
Extending this translation to contexts in a pointwise fashion, we obtain the
following correctness lemma.

\begin{lemma}
  The translation $\shiftTl{-}$ from \CCICPrecs to \CCICPrec forms a syntactic model:
  \begin{enumerate}
  \item If $\Gamma \ccicprecsty t : A$ and $t \redCCIC t'$ in \CCICs then $\shiftTl{\Gamma}
    \ccicprecty \shiftTl{t} \conv \shiftTl{t'} : \shiftTl{A}$ in
    \CCICPrec;
   
  \item  If $\Gamma \ccicprecsty t : A$ then $\shiftTl{\Gamma} \ccicprecty
    \shiftTl{t} : \shiftTl{A}$.
  \end{enumerate}
\end{lemma}

\begin{proof}
For point (1), $\beta$-reduction is preserved thanks to
\nameref{infrule:coe-retr} and \shepherd{the reduction rules for casts can be simulated,
in particular the decomposition of $\Pi$ types into $\?_{[]}$}{all other rules are the same in both systems}.
Point (2) is then immediate from the observation that \nameref{infrule:ccics-prod}
can be translated to an application of \nameref{infrule:cic-prod} followed by
\nameref{infrule:ccic-cum}, \nameref{infrule:cic-abs} is translated to an
application of the same rule followed by \nameref{infrule:ccic-coe}, and
\nameref{infrule:cic-app} is modified with an application of \nameref{infrule:ccic-coeinv}.
\end{proof}

\begin{theorem}[Self-precision of \CCICPrecs embedding]
  \label{thm:ccics-sp}
  If $\ccicprecsty t : A$ then $\sp{\shiftTl{t}}{\shiftTl{A}}$ is derivable.
\end{theorem}
\begin{proof}
  We prove more generally that if $\Gamma \ccicprecsty t : A$ then
  \shepherd{there is some}{we can build a proof} $t'$
  such that
  $\shiftTl{\Gamma}_\varepsilon \ccicprecty t' : \shiftTl{t}_0
  \precision{\shiftTl{A}_0}{\shiftTl{A}_1} \shiftTl{t}_1$, where 
  $\shiftTl{\Gamma, x : A}_\varepsilon :=\shiftTl{\Gamma}_\varepsilon, x_0 :
  \shiftTl{A}_0, x_1 : \shiftTl{A}_1, x_\varepsilon : x_0
  \precision{\shiftTl{A}_0}{\shiftTl{A}_1} x_1$, and $\shiftTl{x}_i := x_i$. 
  The proof proceeds by induction on the typing derivation.
  \shnew{The case of $\lind{A}$ has already been outlined above, thus}
  we only treat the \shnew{other} central case where ${t = \P x : A. B}$.
  By induction hypothesis, we have
  $\shiftTl{\Gamma}_\varepsilon \ccicprecty ih_A : \shiftTl{A}_0
  \precision{\shiftTl{[]_i}_0}{\shiftTl{[]_i}_1} \shiftTl{A}_1$ and
  $\shiftTl{\Gamma}_\varepsilon, x_0 : \shiftTl{A}_0, x_1 : \shiftTl{A}_1,
  x_\varepsilon : x_0 \precision{\shiftTl{A}_0}{\shiftTl{A}_1} x_1 \ccicprecty 
 ih_B : \shiftTl{B}_0 \precision{\shiftTl{[]_i}_0}{\shiftTl{[]_i}_1}
 \shiftTl{B}_1$,
 and need to prove that 
 ${\cum (\P x_0 : \shiftTl{A}_0. \shiftTl{B}_0)}$
   $\precision{\shiftTl{[]_{i+1}}_0}{\shiftTl{[]_{i+1}}_1}$ ${\cum (\P x_1 :
   \shiftTl{A}_1. \shiftTl{B}_1)}$. Hence, using
 \nameref{infrule:univ-prec-def}, \nameref{infrule:cum-prec-cong-ty} and
 \nameref{infrule:cum-prec-unk},
 that $\P x_0 : \shiftTl{A}_0. \shiftTl{B}_0
   \precisionType{i} \P x_1 : \shiftTl{A}_1. \shiftTl{B}_1$.
 The two heterogeneous precision required by \nameref{infrule:pi-prec-cong} are
 direct consequences of $ih_A$ and $ih_B$ using \nameref{infrule:univ-prec-def}
 to relate type and term precision at level $i$.
 Finally, the monotonicity of $\shiftTl{B}_0$ and $\shiftTl{B}_1$ are
 consequences of $ih_B$ and quasi-reflexivity of precision that holds because
 every type in the context is self-precise.
\end{proof}

\shnew{
Combining this theorem with \cref{thm:dgg}, we get that the DGG holds for any \CCICPrecs context.

\begin{corollary}[Dynamic Gradual Guarantee for \CCICPrecs]
  \label{thm:dgg-frag}
  If $\ccicprecsty C : A \to \bool$, then for any $x, y$ such that $x \precision{\shiftTl{A}}{\shiftTl{A}} y$
  is derivable, also $(\coeinv\,\shiftTl{C})~x \precision{\bool}{\bool} (\coeinv\,\shiftTl{C})~y$ is.
\end{corollary}

Terms that fall outside of the \CCICPrecs fragment include
recursive dependent arities such as $\mathtt{nArrow}$ (\cref{sec:internal-prec}), and pathological
terms such as $\Omega$ (\cref{sec:introduction})
that would be non-terminating in a globally gradual system such as \GCICG.
More interestingly, examples like $\mathtt{mult}^{\err}_\listT$ (\cref{sec:exceptions}) can be manually proven to be gradual
even if they do not belong to \CCICPrecs because they use $\catch$ locally. 
}

\section{A model of a reasonably gradual type theory}
\label{sec:model}

In this section we prove \cref{thm:logical-soundness}, that is
the relative consistency of \CCICPrec with a hierarchy of $n$ universes
with respect to \MLTT\footnote{With the standard type formers $\zeroType$, $\oneType$, $\twoType$, W, $\Sigma$, $\mathrm{Id}$ and $\Pi$.} with $(n+1)$ universes and a type of definitionally proof irrelevant propositions.
To do so, we exhibit a model where types are equipped
with a relation reflecting precision.
We formalized the components of this model 
(for two universes $[]_{0}$ and $[]_{1}$)
in Coq. 
%
The construction of the model can be stratified in 3 layers:
\begin{itemize}
\item first, a
  computational layer that provides meaning to casts and exceptional terms $\err_A,
  ?_A$ ;
\item second, a relational layer that equips every type with a relation and
  defines a compatible global heterogeneous relation between elements;
\item third, a logical layer ensuring that said relations do capture
  well-behaved casts whenever all inputs are adequately related.
\end{itemize}

\paragraph{Computational layer}
The computational layer closely resembles the discrete model
of~\citet{lennonAl:toplas2022}, and we explain here its main features.
The introduction of exceptional terms follows the approach of
\ExTT~\cite{pedrotTabareau:esop2018}.
Its main point is to extend each inductive type with two new constructors, one
for $\?$ and one for $\err$.
Product types and functions are left unmodified, defining
$\?$ and $\err$ pointwise.
%

%
We depart from this model on universes, so that we can define the cast primitive
by case analysis on types.
Taking inspiration from~\citet{BoulierPT17}, we interpret types
as \emph{codes} when they are seen as terms, and as \emph{the semantics of those
  codes} when they are seen as types.
Thus, the standard interpretation for a term inhabiting a type is maintained,
but a function taking as argument an element of the universe $\Utsl{}_i$ can now
perform a case analysis on the code of the type.
The precise construction of the interpretation of the universe hierarchy
employs a technique presented by \citet{SattlerV20}.
We first define an inductive family $\Elcode : \univ{i} \to
\univ{i+1}$ describing codes for types and then pack it as $\Utsl_i :
\univ{i+1} := \Sigma(A : \univ{i})~\Elcode\,A$, using the first projection as decoding.
We can then define an operation 
\coqe$cast : forall (A B : Utsli), A -> B$
by induction on these codes, following the reduction rules
of~\cref{fig:ccic-reduction}.
\cref{fig:univ-code} presents a simplified version of this construction, to which
codes for the translation of the types $\err_{[]_i}$, $\?_{[]_i}$,
$\prop$, $\listT\,A$ and $[]_j$ (for $j < i$) are added in the actual development.
%
\begin{figure}
\begin{minipage}[t]{0.4\linewidth}
\begin{coqsmall}
Let El X : Type := X.1.

Inductive code : Typei -> Typesuci :=
| code_Nat : code Nat
| code_Pi (A : Utsli) (B : El A -> Utsli) :
  code (forall a, El (B a))
| ... 
where Utsli := (Sigma(A : Typei) code A).
\end{coqsmall}
\end{minipage}%
\begin{minipage}[t]{0.6\linewidth}
\begin{coqsmall}
Fixpoint cast (A B : Utsli) : A -> B :=
  match A.2, B.2 with
  | code_Nat, code_Nat => fun n => n
  | code_Nat, code_Pi _ _ => fun _ => err _
  | code_Pi A0 A1, code_Nat => fun _ => err _
  | code_Pi A0 A1, code_Pi B0 B1 =>
    fun (f : forall a, El (A1 a)) (b : El B0) => 
    cast (A1 _) (B1 b) (f (cast B0 A0 b))
  | ... .
\end{coqsmall}
\end{minipage}

  \caption{Simplified code for the universe of codes and of cast.}
  \label{fig:univ-code}
\end{figure}

The exceptional model of \citet{pedrotTabareau:esop2018} leave the
interpretation of exceptions at the universe $[]$ unspecified.
We exploit this underspecification, and define $\err_{[]_i}$
as the unit type $\oneType$ with a single element.
%
$\?_{[]_{i+1}}$ is interpreted by an inductive type \coqe$unknown$
(\cref{fig:implementation-precision-unknown}) closed by all
type constructors but dependent functions.
Beyond the two constructors \coqe$err_unknown$ and \coqe$unk_unknown$
interpreting respectively $\err_{\?_{[]}}$ and $\?_{\?_{[]}}$,
\coqe$univ_unknown$ allows to embed the preceding universe, \coqe$cum_unknown$
hosts any type from said preceding universe (including product types), and
\coqe$list_unknown$ can be used to embed lists of elements from $\Utsl_{i+1}$.
Additional inductive types would be represented with supplementary constructors.
The interpretation of $\?_{[]_0}$ do not use \coqe$univ_unknown$ and \coqe$cum_unknown$.
\begin{figure}
\begin{minipage}[t]{0.4\linewidth}
\begin{coqsmall}
Inductive unknown :=
| err_unknown 
| univ_unknown (A : Utsli) 
| cum_unknown (A : Utsli) (a : El A) 
| list_unknown (l : list unknown) 
| unk_unknown. 
\end{coqsmall}
\end{minipage}%
\begin{minipage}[t]{0.6\linewidth}
\begin{coqsmall}
Inductive prec_unk : unknown -> unknown -> SProp :=
| err_any : sp x -> prec_unk err_unknown x
| unk_any : sp x -> prec_unk x unk_unknown
| univ_prec : A preciseUniv B ->
  prec_unk (univ_unknown A) (univ_unknown B)
| cum_prec A a B b : sp A -> sp B -> hprec A B a b ->
  prec_unk (cum_unknown A a) (cum_unknown B b)
| list_prec l1 l2 : lift_list prec_unk l1 l2 ->
  prec_unk (list_unknown l1) (list_unknown l2).
\end{coqsmall}
\end{minipage}
  \caption{Translation of $\?_{[]}$ and its precision relation.}
  \label{fig:implementation-precision-unknown}
\end{figure}

\paragraph{Relational layer}
We now endow the translation of every type with a homogeneous relation
\coqe$prec : forall (A : Utsli), A -> A -> SProp$.
Thanks to the characterization of heterogeneous precision in
\cref{fig:ccic-prec-properties}, we can use \coqe$prec$ together with
\coqe$cast$ to obtain an heterogeneous relation on all types at the same
universe level:
\begin{coq}
Let hprec (A B : Utsli) (a : El A) (b : El B) : SProp := prec A a (cast A B b) /\ prec B b b.
\end{coq}
The construction of \coqe$prec$ proceeds first by induction on the universe
level, and then by induction on the code of the type.
The cases for $\err_{[]_i}$, $\prop$, inductive types, dependent functions and
cumulativity injection follow the formulae given for precision in~\cref{fig:ccic-prec}. In particular,
defining homogeneous precision at function types relies on heterogeneous precision
on the codomain.
%
On universes, we use precision for the smaller universe, obtained by
induction hypothesis on the universe level.
The precision for \coqe$unknown$ is described on the right
of~\cref{fig:implementation-precision-unknown}.
\coqe$err_unknown$ and \coqe$unk_unknown$ are respectively smaller and larger
than self-precise terms of any summand.
\coqe$univ_prec$ embeds the relation from $\Utsl_i$ and \coqe$cum_prec$ relate
elements of self-precise types using the heterogeneous relation determined by
$\Utsl_i$.
Finally, \coqe$list_prec$ lifts the precision on \coqe$unknown$ to lists.

\paragraph{Property layer}
Once all definitions are in place, we need to show that the relations thus defined
do characterize well-behaved casts. This is summarized by the following definitions.
\begin{definition}[Partial preorder] A partial preorder on a type $X$ is a transitive
  and quasi-reflexive relation $\precision{X}{X}$ on $X$.
\end{definition}
An element $x$ of a partial preorder $X$ is self-precise, notation $\sp{x}{X}$,
when $x \precision{X}{X} x$.  
A pair of functions $f : X \to Y$, $g : Y \to X$ between partial preorders $X,Y$ forms an
\emph{embedding projection pair}
if it satifisfies the condition of~\cref{def:eppair}.
A \emph{type family with casts} consists of a type family $B : A \to \univ{}$ equipped
with two functions $\upcast^B_{a,a'} : B\,a \to B\,a'$ and $\downcast^B_{a,a'} :
B\,a' \to B\,a$.
\begin{definition}[Indexed partial preorder]
  If $A$ is a partial preorder and $B$ a type family with cast such that each $B\,a$ is
  endowed with a relation $\precisionDiag{B\,a}$, then $B$ is an indexed partial
  preorder when
  \begin{itemize}
  \item whenever $\sp{a}{A}$, $\precisionDiag{B\,a}$ is a partial preorder;
  \item if $a \precisionDiag{A} a'$, then $(\upcast^B_{a,a'},
    \downcast^B_{a,a'})$ forms an \eppair;
  \item whenever $\sp{a}{A}$, $\sp{b}{B\,a}$, $\upcast^B_{a,a} b
    \equiprecision{B\,a} b \equiprecision{B\,a}\downcast^B_{a,a} b $;
  \item if $a_0 \precisionDiag{A} a_1 \precisionDiag{A} a_2$, $\sp{b}{B\,a_0}$,
    then $\upcast^{B}_{a_1,a_2}\upcast^B_{a_0,a_1} b \equiprecision{B\,a_2}
    \upcast^B_{a_0,a_2} b$ and \mbox{$\downcast^B_{a_2,a_1}\downcast^{B}_{a_1,a_0} b
    \equiprecision{B\,a_2} \downcast^B_{a_2,a_0} b$}
  \end{itemize}
\end{definition}

Now the model validates the following:

\begin{theorem}[Properties of precision]
  \label{thm:properties-precision}
  The universes $\Utsl_i$ is a partial preorder for term and type precision and
  the type families $\El : \Utsl_i \to \univ{i}$ equiped with \coqe$cast$ are
  indexed partial preorders. \coqe$unknown$ is a greatest element for term
  precision on the universe.
\end{theorem}

The proof of this theorem proceed by induction on multiset of codes, showing
that the relation induced by a code is partial preorder, that pairs of casts
between the partial preorders induced by a pair of codes form an \eppair and
that the eppairs induced by a triple of code compose adequately.
To that end, we prove and use a lemmas asserting that type constructors from
$\Utsl_i$ preserve partial preorders and \eppairs, e.g. the relation on $\forall
a : \El\, A, \El (B\, a)$ induced by a code \coqe$code_Pi A B$ is a partial
preorder whenever $\El\,A$ is a partial preorder and $\El \circ B$ is an indexed partial
preorder.

The properties presented in~\cref{sec:precision-properties} are consequences of
this theorem, using the decomposition of heterogeneous term relation
through any upper bound for the precision relation, the fact that any type at
universe level $i$ is bounded by $\?_{[]_{i+1}}$ and cumulativity preserves and
reflects precision.

\shnew{
\paragraph{Metatheoretical properties induced by the model}
Since $\prop$ is translated to \coqe$SProp$ and $\bot$ to $\bot$ in the model, any closed proof
$\vdash e : \bot$ induces a corresponding closed term of an empty type in the target
type theory.
This proves the relative consistency of \CCICPrec with respect to
 \MLTT equipped with enough universes and extended with a type of
strict proposition as claimed in \cref{thm:logical-soundness}.
This result can be further refined by analyzing the translation of each
reduction steps from \cref{fig:ccic-reduction} and realizing that these can be
simulated by at least one step in the target type theory, reusing a proof
technique found in \citet{lennonAl:toplas2022}.

\begin{theorem}[Normalization of \CCICPrec]
 \CCICPrec is normalizing.
\end{theorem}

\begin{proof}[Proof sketch]
  Since each step of reduction in the source is mapped to at least one step of
  reduction in the target, any infinite reduction sequence in the source maps to
  an infinite reduction sequence in the target as well.
  \citet{gilbert:hal-01859964} show that MLTT+\coqe$SProp$ is normalizing, so an
  infinite reduction sequence cannot exist in the target, and so not in the
  source either.
\end{proof}
}

\section{Extensions of \CCICPrec}
\label{sec:impl-deta-agda}

We now  discuss several extensions of \CCICPrec for future work.

\subsection{Observational Equality}
\label{sec:equality}

\CCICPrec features two kinds of sorts, $[]$ for (impure) computationally relevant types and $\prop$ for
definitionally proof irrelevant propositions.
The main purpose of $\prop$ is to be able to define precision
internally in \CCICPrec, by induction on types.
In the recent work of \citet{pujet:hal-03367052}, $\prop$ is used in the
same way to define a notion of observational equality by induction on types,
satisfying extensionality principles.
It turns out that internal precision and observational equality can
both be integrated in \CCICPrec.
\shnew{We can add in $\prop$ a notion of equality $x =_A y$ for any
terms $x$ and $y$ of type $A$,
together with a transport operation:
  \begin{mathpar}
    \inferrule[Transport]{\Gamma{} \vdash A : []_i \\ \Gamma{} \vdash B : []_i \\ \Gamma{} \vdash e : A =
      B \\ \Gamma{} \vdash t : A}
    {\Gamma{} \vdash \transport{A}{B}{e}{t} : B}
    \ilabel{infrule:transport}
  \end{mathpar}
%
%
Intuitively, transport can be seen as the \emph{safe} version
of cast, using a proof of equality between types in the logical layer
as a guard to ensure it never fails.}


There are two main interests in adding a notion of observational
equality to \CCICPrec.
First, it allows us to state many properties than cannot be only
stated using internal precision. For instance, equality is necessary
to express internally what antisymmetry means for internal precision,
and prove that it holds on types for which all terms are
self-precise.
Second, it provides a canonical way to express (non-gradual) subset
types in \CCICPrec, thus recovering a flavor of indexed inductive
types.

\subsection[Inductive Types]{Inductive Types}
\label{sec:general-inductive-types}

\shepherd{}{
A large class of inductive types can be encoded using well-founded trees
$\mathbb{W}\,A\,B$ with nodes indexed by $A: []$ of arity $B : A \to []$, \aka
$\mathbb{W}$-types~\cite{AltenkirchGHMM15,Hugunin20}.
A mild extension of \CCICPrec could add such types $\mathbb{W}\,A\,B$ with a
constructor $\texttt{sup}_{A,B} : \P(a : A)(k : B\,a \to \mathbb{W}\,A\,B) \to
\mathbb{W}\,A\,B$ and a corresponding eliminator $\catch_{\mathbb{W}\,A\,B}$.
These types would then be self-precise whenever $A$ is a self-precise type and $B$
self-precise as a type family.
In general, it is not reasonable to expect $\mathbb{W}$-types to be below unknown, that
is $\mathbb{W}\,A\,B \precisionDiag{[]_i} \?_{[]_i}$, because the constructor
\texttt{sup} takes a function as argument that cannot be faithfully encoded in
$\?_{[]_i}$.
The more restricted class of \emph{finitary} $\mathbb{W}$-types, meaning that
$B\,a$ is a finite type for any $a : A$, however supports such a
bounding rule so that finitary $\mathbb{W}$-types are also precise as terms in
the universe.
The inductive type of lists $\listT\,X$ is an instance of a finitary
$\mathbb{W}$-type with $A = 1 + X$, $B(\texttt{inl}\,\unitK) = \zeroType$
and $B(\texttt{inr}\,x) = \unit$.
} 
It is also possible to add general indexed inductive types in
$\prop$, such as less-or-equal in $\nat$.
\citet{gilbert:hal-01859964} describe a general criterion to detect
which inductive types in $\prop$ can be eliminated into $[]$. Basically,
this criterion amounts to detecting when an indexed inductive type can be
encoded with a fixpoint over its indices.
This criterion also works for \CCICPrec, and could be reused
directly.

\subsection{From \CCICPrec to a Gradual Proof Assistant}

\CCICPrec is still quite far from a real-life proof assistant.
As explained at the beginning of \cref{sec:action},
usual gradual systems are separated into two languages:
a source language where types are compared in an optimistic
way using the wildcard $\?$, and a target language with casts to
explicitely flag where those optimistic assumptions are made,
so as to be able to raise errors in case of type incompatibilities
discovered during program evaluation.
Here we concentrated on designing the target language, as our contributions
apply mostly to it, with
the expectation that the source and elaboration layers as presented in
~\citet{lennonAl:toplas2022} could be easily adapted to our extensions.
Consequently, we chose to present our type theory in a standard,
undirected fashion, rather than using the
bidirectional approach of~\citet{lennonAl:toplas2022}.
However, building an actual proof assistant involves tackling
that elaboration layer, and the many subtle points it involves, which were only
partially solved by~\cite{lennonAl:toplas2022}. One example would be the
interaction between unification (the main and crucial feature of elaboration
in \textit{e.g.} Coq) and gradual features of the language, especially consistency.
But even if one considers only the target language, incorporating it in an
actual proof assistant is no small feat. In \CCICPrec, we made a wealth of
technical choices (impredicativity of $\prop$, explicit cumulativity, and so on)
that might need to be reconsidered if one wishes to integrate
gradual features in a proof assistant that takes a different path.
In particular, a proper treatment of universe levels is a challenge.
For instance, a system more flexible (and probably easier to use) than \CCICPrec
would allow casts between types at different levels, but this would cause
an unprecedented dependency between reduction (of casts)
and universe levels, which in turn raises subtle implementation questions.

Similarly, we made some choices in the definition of precision, both in 
the rules of~\cref{fig:ccic-prec} and the properties reflected in \CCICPrec in
\cref{sec:precision-properties}. They were in part guided by the aim to make
the system as ready for use as possible,
but they might need to be reconsidered in a practical implementation.

Finally, an interesting design point pertains to the $\catch$ primitive.
Actual proof assistants usually do not rely on recursors, but instead
provide facilities for pattern-matching in various forms. Implementations
of $\catch$ should be adapted to those. In particular, a mechanism to present
monotone $\catch$ as presented in \cref{sec:prec-monotone-frag} could take
inspiration from the implementation of higher inductive types, with
path-constructors replaced by monotonicity constraints.

\section{Related Work}
\label{sec:related-work}

\emph{Effects in dependent type theory.}
Incorporating effects in type theory, specifically errors as needed for gradual 
systems, is particularly challenging.
Indeed, the presence of effects triggers a strong
tension with the metatheoretic properties of \CIC, putting logical consistency in danger,
as clarified by the Fire Triangle of~\citet{pedrotTabareau:popl2020}.
Several programming languages mix dependent types with effectful
computation, either giving up on metatheoretical properties, such as
Dependent Haskell~\cite{eisenberg2016dependent}, which allows diverging type-level expressions,
or by restricting the
dependent fragment to pure
expressions~\cite{xiPfenning:pldi98,swamyAl:popl2016}.
\shnew{\citet{stumpAl:par2010} study the sound coexistence of a type theory with diverging terms via an effect system and a mechanism of termination casts to recover totality for any term given a proof of its termination. This mechanism is used in Trellys~\cite{kimmellAl:plpv2012} and its successor Zombie~\cite{CasinghinoAl:popl2014}, which are call-by-value dependently-typed languages that separate the pure logical fragment from the impure programming fragment using consistency classifiers in the typing judgment. This integrated approach supports sound reasoning about potentially diverging programs.}
Recently, \citet{PedrotT17,pedrotTabareau:esop2018}
build up from general considerations on effects to specifically
consider exceptions in type theory. \citet{pedrotAl:icfp2019} introduce \RETT, 
exploiting universe hierarchies to introduce a separation between an effectful, inconsistent layer
and a pure, consistent one to reason about the effectful one. \CCICPrec is directly inspired by \RETT to support sound reasoning about gradual programs.

\emph{Strict propositions and observational equality.}
It has long been recognized that equality in standard \MLTT is too syntactic. 
Observational type theory~\cite{altenkirchAl:plpv2007} was proposed to address this issue, 
but only thanks to work on incorporating (definitional) irrelevance in dependent type theory \cite{Abel2012,gilbert:hal-01859964} was it possible to recently turn this proposition into a
concrete system~\cite{pujet:hal-03367052}, by using the definitionally proof-irrelevant sort to host the observational equality. The sort $\prop$ and the
precision relation of \CCICPrec are very much inspired respectively by the
sort of definitionally proof-irrelevant propositions of \citet{gilbert:hal-01859964}
and the observational equality of \citet{pujet:hal-03367052}.

\shepherd{}{
\emph{Directed type theory}.
Segal and Rezk types characterize well-behaved types in directed type theory
\cite{weaverConstructiveModelDirected2020,riehlTypeTheorySynthetic} in a fashion
very similar to self-precise types in \CCICPrec: Segal types have
(up-to-homotopy) unique composition of
morphisms (transitivity), while Rezk types satisfy a local notion of univalence
(antisymmetry).
In these works, any type is equipped with (higher) identities, an important
difference with our setting where we do not globally ensure reflexivity of the precision
relations, that is self-precision of types and terms.
} 

\emph{Gradual typing and dependent types.} 
This work continues a line of research in combining dependent types and dynamic type checking, as first explored by~\cite{ouAl:tcs2004}, more specifically following the gradual typing approach~\cite{siekTaha:sfp2006,siekAl:snapl2015}, and extending it to a full-blown dependent type theory.
\citet{ouAl:tcs2004} study a programming language with separate dependently- and simply-typed fragments, using arbitrary runtime checks at the boundary. The blame calculus of \citet{wadlerFindler:esop2009} considers subset
types on base types, where the refinement is an arbitrary term, as in
hybrid type checking \cite{knowlesFlanagan:toplas2010}, but lacks dependent function types. 
\citet{tanterTabareau:dls2015} provide casts for subset types with decidable properties in \Coq, 
and \citet{dagandAl:jfp2018} support dependent interoperability~\cite{oseraAl:plpv2012} in \Coq. All these approaches lack the notion of precision that is central to gradual typing.
Gradual refinement types~\cite{lehmannTanter:popl2017} differ from the gradual subset types presented here in that they are an extension of liquid types~\cite{liquid:popl2008} with imprecise logical formulas, based on an SMT-decidable logic about base types. 
\citet{eremondiAl:icfp2019} study the gradualization of \CC, and propose approximate normalization to ensure decidable typechecking. Approximate normalization satisfies the dynamic gradual guarantee, but not graduality in the sense of \cite{newAhmed:icfp2018}, because casting to an imprecise type and back can yield the unknown term instead of the original term.
The most recent and complete attempt to gradualize \CIC, upon which we build in this work, is the study of \GCIC and its underlying cast calculus \CCIC~\cite{lennonAl:toplas2022}, which comes under three variants. \CCICPrec is an extension of \CCIC that allows for sound reasoning about gradual programs and, thanks to internal precision, can account for the specific form of graduality supported by  \CCICN, the normalizing conservative extension of \CIC, and can embed \CCICs as a subclass of terms that are self-precise. 
\shnew{\citet{eremondiAl:icfp2022} extends \GCIC with gradual propositional equality using runtime witnesses of plausible equality, taking inspiration from evidence tracking in Abstracting Gradual Typing~\cite{garciaAl:popl2016}.}





\bibliography{strings,pleiad,bib,local,common}

\end{document}

\appendix

\section{Congruence rules for \CCICPrec reduction}
\label{sec:ccic-red-congruence}

\Cref{fig:ccic-reduction} omits congruence rules for space considerations, they are given
in \cref{fig:ccic-red-congruence}. The term constructors $\listT$, $\nilK[A]$, $\consK[A]$,
$\cum$, $\coe$, $\coeinv$, $\lcatch{A}$, $\Box$ and $\boxcatch{P}$ behave like constants applied
to arguments, \eg $\listT\,A \redCCIC \listT\,A'$ whenever $A \redCCIC A'$.

\begin{figure}[h]
  \begin{mathpar}
    \inferrule{A \redCCIC A'}{\P x : A.\ B \redCCIC \P x : A'.\ B} \and
    \inferrule{B \redCCIC B'}{\P x : A.\ B \redCCIC \P x : A.\ B'} \and
    \inferrule{A \redCCIC A'}{\l x : A.\ t \redCCIC \l x : A'.\ t} \and
    \inferrule{t \redCCIC t'}{\l x : A.\ t \redCCIC \l x : A.\ t'} \and
    \inferrule{t \redCCIC t'}{t\ u \redCCIC t'\ u} \and
    \inferrule{u \redCCIC u'}{t\ u \redCCIC t\ u'} \and
    \inferrule{A \redCCIC A'}{\nilK[A] \redCCIC \nilK[A']} \and
    \inferrule{A \redCCIC A'}{\consK[A] \redCCIC \consK[A']} \and
    \inferrule{A \redCCIC A'}{\?_{A} \redCCIC \?_{A'}} \and
    \inferrule{A \redCCIC A'}{\err_{A} \redCCIC \err_{A'}} \\
    \inferrule{A \redCCIC A'}{\cast{A}{B}{t} \redCCIC \cast{A'}{B}{t}} \and
    \inferrule{B \redCCIC B'}{\cast{A}{B}{t} \redCCIC \cast{A}{B'}{t}} \and
    \inferrule{t \redCCIC t'}{\cast{A}{B}{t} \redCCIC \cast{A}{B}{t'}} \\
    \inferrule{A \redCCIC A'}{\lcatch{A}\,P\,h_{\nilK}\,h_{\consK}\,h_{\err}\,h_\?\,l \redCCIC \lcatch{A'}\,P\,h_{\nilK}\,h_{\consK}\,h_{\err}\,h_\?\,l} \and
    \inferrule{A \redCCIC A'}{\forall (x : A),P \redCCIC \forall (x : A'),P} \and
    \inferrule{P \redCCIC P'}{\forall (x : A),P \redCCIC \forall (x : A),P'} \and
    \inferrule{P \redCCIC P'}{\boxcatch{P} A\,h\,h_{\err}\,h_\?\,l \redCCIC \boxcatch{P'} A\,h\,h_{\err}\,h_\?\,l} \and
  \end{mathpar}

  \caption{Congruence rules omitted in \cref{fig:ccic-reduction}}
  \label{fig:ccic-red-congruence}
\end{figure}

\end{document}

\appendix

\begin{figure}
\begin{mathpar}
  \inferrule[List-Catch-Mon]
  { \Gamma \vdash P : \listT\,A \to []_i\\
    \Gamma \vdash h_{\nilK} : P\,\nilK[A]\\
    \Gamma \vdash h_{\consK} : \P (a : A)(l : \listT\,A) . P\,l {\to} P(\consK[A]\,a\,l)\\
    \Gamma \vdash h_{\err} : P\,\err_{\listT~A}\\
    \Gamma \vdash h_{\?} : P\,\?_{\listT~A}\\
    \Gamma \vdash e_{\err,\nilK} : h_{\err} \precisionDiag{\listT\,A} h_{\nilK} \\
    \Gamma \vdash e_{\nilK,\?} : h_{\nilK} \precisionDiag{\listT\,A} h_{\?} \\
    \Gamma \vdash e_{\err,\consK} : \forall (a : A)(e_a : \sp{a}{A})(l :
    \listT\,A)(e_l : \sp{l}{\listT\,A})(ih : P\,l) (e_{ih} : \sp{ih}{P\,l}),\hspace{3cm}\\
    \hspace{6cm}\err_{\listT\,A} \precision{P\,\err_{\listT\,A}}{P(\consK[A]\,a\,l)} h_{\consK}\,a\,l\,ih\\
    \Gamma \vdash e_{\consK,\?} : \forall (a : A)(e_a : \sp{a}{A})(l :
    \listT\,A)(e_l : \sp{l}{\listT\,A})(ih : P\,l) (e_{ih} : \sp{ih}{P\,l}),\hspace{3.2cm}\\
    \hspace{6cm}h_{\consK}\,a\,l\,ih \precision{P(\consK[A]\,a\,l)}{P\,\?_{\listT\,A}}\?_{\listT\,A}\\
    \Gamma \vdash e_{\consK} : \forall (a~a' : A)(e_a : a \precisionDiag{A} a')
    (l~l' : \listT\,A)(e_l : l \precisionDiag{\listT\,A} l')\hspace{4cm}\\
    \hspace{1cm}(ih : P\,l) (ih' : P l') (e_{ih} : ih \precisionDiag{P\,l} ih'),~
    h_{\consK}\,a\,l\,ih \precision{P(\consK[A]\,a\,l)}{P(\consK[A]\,a'\,l')}h_{\consK}\,a'\,l'\,ih'\\
  }
  {\Gamma \vdash \wit : \sp{(\lcatch{A}\,P\,h_{\nilK}\,h_{\consK}\,h_{\err}\,h_\?)}{\forall l, P\,l}}
  \ilabel{infrule:ccic-list-catch-mon}
\end{mathpar}
  
  \caption{Monotonicity rule for catch on lists}
  \label{fig:monotonicity-catch}
\end{figure}

%% file: definitions.tex
\usepackage{anyfontsize} 
\usepackage{amsmath}     
\usepackage{amssymb}     
\usepackage{amsbsy}      
\usepackage{semantic}    
\usepackage{stmaryrd}   
\usepackage{mathpartir}  
\usepackage{braket}      
\usepackage{mathtools}   
\usepackage{thmtools}    
\usepackage{hyperref}
\usepackage[capitalise]{cleveref}
\usepackage{xspace}
\usepackage{mathrsfs} 
\usepackage{proof}
\usepackage{rotating}
\usepackage{tikz-cd}
\usepackage{xcolor}
\usepackage{changepage}
\usepackage{centernot}
\usepackage{comment}
\usepackage{listings}
\usepackage{centernot}
\usepackage{multirow}
\usepackage{xparse}
\usepackage{pgothic} 
\usepackage{tensor} 
\usepackage{bbold} 
\usepackage{subcaption}
\usepackage{tcolorbox}

\usepackage{agda-syntax}

\colorlet{darkblue}{RoyalBlue}
\colorlet{darkgreen}{green!45!black}
\definecolor{darkred}{rgb}{0.55, 0.0, 0.0}
\definecolor{darkviolet}{rgb}{0.58, 0.0, 0.83}
\definecolor{lightblue}{rgb}{0.68, 0.85, 0.9}
\usepackage{lstcoq}

\makeatletter
\newcommand{\xMapsto}[2][]{\ext@arrow 0599{\Mapstofill@}{#1}{#2}}
\def\Mapstofill@{\arrowfill@{\Mapstochar\Relbar}\Relbar\Rightarrow}
\makeatother

\definecolor{propcolor}{HTML}{3F7D31}
\definecolor{mypurple}{HTML}{5B069D}
\definecolor{mistyrose}{rgb}{1.0, 0.89, 0.88}
\definecolor{pastelgray}{rgb}{0.81, 0.81, 0.77}

\newtheoremstyle{theoremthin}
{\smallskipamount}
{\smallskipamount}
{\itshape}
{0pt}
{\scshape}
{.}
{ }
{}

\theoremstyle{theoremthin}

\newtheorem{theorem}{Theorem}
\newtheorem{lemma}[theorem]{Lemma}
\newtheorem{proposition}[theorem]{Proposition}
\newtheorem{corollary}[theorem]{Corollary}
\newtheorem{definition}{Definition}

\theoremstyle{remark}

\makeatletter
\renewenvironment{proof}[1][\proofname]{\par
  \pushQED{\qed}%
  \normalfont
  \topsep0pt \partopsep0pt 
  \trivlist
\item[\hskip\labelsep
  \itshape
  #1\@addpunct{.}]\ignorespaces
}{%
  \popQED\endtrivlist\@endpefalse
  \addvspace{2pt plus 2pt} 
}
\makeatother

\newcommand{\?}{{\ensuremath{\operatorname{\boldsymbol{?}}}}} 
\mathlig{[]}{\square}
\mathlig{.1}{._{1}}
\mathlig{.2}{._{2}}

\usepackage{scalerel}
\newcommand{\axiom}{\texttt{ax}}

\newcommand{\pgrad}{\ensuremath{\mathcal{G}}\xspace}

\newcommand{\pnorm}{\ensuremath{\mathcal{N}}\xspace}

\newcommand{\CC}{\ensuremath{\mathrm{CC}_{\omega}}\xspace}
\newcommand{\Coq}{\ensuremath{\mathrm{Coq}}\xspace}
\newcommand{\CIC}{\ensuremath{\mathsf{CIC}}\xspace}
\newcommand{\MLTT}{\ensuremath{\mathsf{MLTT}}\xspace}

\newcommand{\CICs}{\ensuremath{\mathsf{CIC}^{\uparrow}}\xspace}
\newcommand{\ExTT}{{\ensuremath{\mathsf{ExTT}}}\xspace}

\newcommand{\RETT}{\ensuremath{\mathsf{RETT}}\xspace}
\newcommand{\TTOBS}{\ensuremath{\mathsf{TT^{obs}}}\xspace}

\newcommand{\GCIC}{\ensuremath{\mathsf{GCIC}}\xspace}
\newcommand{\GCICG}{\ensuremath{\GCIC^{\pgrad}}\xspace}
\newcommand{\GCICN}{\ensuremath{\GCIC^{\pnorm}}\xspace}
\newcommand{\GCICs}{\ensuremath{\GCIC^{\uparrow}}\xspace}
\newcommand{\CCIC}{\ensuremath{\mathsf{CastCIC}}\xspace}
\newcommand{\CCICG}{\ensuremath{\CCIC^{\pgrad}}\xspace}
\newcommand{\CCICPrec}{\ensuremath{\mathsf{GRIP}}\xspace}
\newcommand{\CCICN}{\ensuremath{\CCIC^{\pnorm}}\xspace}
\newcommand{\CCICs}{\ensuremath{\CCIC^{\uparrow}}\xspace}
\newcommand{\CCICPrecs}{\ensuremath{\mathsf{GRIP}^{\uparrow}}\xspace}

\newcommand{\Agda}[0]{Agda\xspace}

\newcommand{\ie}{\emph{i.e.,}\xspace}

\newcommand{\eg}{\emph{e.g.,}\xspace}

\newcommand{\aka}{a.k.a.\xspace}

\renewcommand{\terms}{\operatorname{Term}}
\newcommand{\types}{\operatorname{Whnf}_{[]}}
\renewcommand{\H}[0]{\operatorname{Head}}

\renewcommand{\P}{\operatorname{\Pi}}

\renewcommand{\l}{\operatorname{\lambda}}
 
\newcommand{\conv}{\equiv}

\newcommand{\cum}{\mathop{\iota}}
\newcommand{\coe}{\mathop{\uparrow}\! }
\newcommand{\coeinv}{\mathop{\downarrow} }

\newcommand{\redCCIC}[0]{\leadsto}
\newcommand{\transRel}[1]{\operatorname{{#1}^\ast}}
\newcommand{\rtred}{\transRel{\redCCIC}}

\newcommand{\pre}{\sqsubseteq}
\newcommand{\gpre}{\sqsupseteq}
\newcommand{\equiprecision}[1]{\sqsupseteq\!\sqsubseteq_{#1}}
\newcommand{\equiprecise}{\equiprecision{}}

\newcommand{\obsEquiv}{\approx^{obs}}
\newcommand{\obsApprox}{\preccurlyeq^{obs}}

\newcommand{\ind}{\mathtt{ind}}
\newcommand{\catch}{\mathtt{catch}}
\newcommand\catchat[1]{\catch_{#1}}

\newcommand{\bool}{\mathbb{B}}
\newcommand{\btrue}{\mathtt{true}}
\newcommand{\bfalse}{\mathtt{false}}

\newcommand{\nat}{\mathbb{N}}

\newcommand{\unit}{\mathtt{unit}}
\newcommand{\unitK}{\ensuremath{\mathtt{()}}}

\newcommand\zeroType[0]{\mathbb{0}}
\newcommand\oneType[0]{\mathbb{1}}
\newcommand\twoType[0]{\mathbb{2}}

\newcommand{\listT}{\mathbb{L}}
\DeclareDocumentCommand\nilK{o}{\mathtt{nil}\IfValueT{#1}{_{#1}}}
\DeclareDocumentCommand\consK{o}{\mathtt{cons}\IfValueT{#1}{_{#1}}}

\newcommand{\append}{\mathop{+ \! +}}

\newcommand{\ascdom}[3]{\langle {#3} {\ \Leftarrow \ } {#2} \rangle\,{#1}}

\newcommand{\cast}[3]{\ascdom{#3}{#1}{#2}}
\newcommand{\castmid}[4]{\ascdom{\!\ascdom{#4}{#3}{#2}}{#2}{#1}}

\newcommand{\ascName}{\mathrel{::}}
\newcommand{\asc}[2]{#1\ascName#2}

\newcommand{\err}{\operatorname{\mathtt{err}}}

\newcommand\ccicprecsty{\vdash_{\CCICPrecs}}
\newcommand\ccicprecty{\vdash_{\CCICPrec}}

\newcommand{\ccicty}{\operatorname{:}}

\newcommand{\Utsl}[0]{\square \!\!\!\! \square}
\newcommand{\El}[0]{\mathrm{El}}

\newcommand{\univ}[1]{\square_{#1}}

\newcommand{\Elcode}[0]{\text{code}}
\newcommand{\code}[1]{\widehat{#1}}

\newcommand{\hd}[1]{\mathrm{head}\,#1}

\colorlet{darkred}{red!80!gray}

\usepackage{float}
\restylefloat{figure}

\makeatletter
\def\mathcolor#1#{\@mathcolor{#1}}
\def\@mathcolor#1#2#3{%
	\protect\leavevmode
	\begingroup\color#1{#2}#3%
\endgroup
}
\makeatother

\definecolor{grey}{rgb}{0.5, 0.5, 0.5}

\makeatletter
\def\slashedarrowfill@#1#2#3#4#5{%
  $\m@th\thickmuskip0mu\medmuskip\thickmuskip\thinmuskip\thickmuskip
  \relax#5#1\mkern-7mu%
  \cleaders\hbox{$#5\mkern-2mu#2\mkern-2mu$}\hfill
  \mathclap{#3}\mathclap{#2}%
  \cleaders\hbox{$#5\mkern-2mu#2\mkern-2mu$}\hfill
  \mkern-7mu#4$%
}
\def\rightslashedarrowfill@{%
  \slashedarrowfill@\relbar\relbar\mapstochar\rightarrow}
\newcommand\xslashedrightarrow[2][]{%
  \ext@arrow 0055{\rightslashedarrowfill@}{#1}{#2}}
\makeatother

\newcommand{\boxedrule}[1]{\flushleft \fbox{#1} \vspace{-1em}}

\newcommand{\eppair}[0]{ep-pair\xspace}
\newcommand{\eppairs}[0]{\eppair{}s\xspace}

\newcommand\preceppair{{\sqsubseteq}\text{-}\texttt{\eppair}}

\newcommand{\upcast}[0]{\Uparrow}
\newcommand{\downcast}[0]{\Downarrow}

\newif\ifappendix\appendixtrue

\newif\ifleftovers\leftoversfalse

\makeatletter

\DeclareDocumentCommand \inferrule { s O{} m m}{%
	\IfBooleanTF{#1}%
	{%
		\mpr@inferstar[#2]{#3}{#4}%
	}{%
		\mpr@inferrule[#2]{#3}{#4}%
	}%
	\IfValueT{#2}%
	{%
		\my@name@inferrule{#2}%
	}%
}
\NewDocumentCommand \my@name@inferrule { m }{%
	\def\@currentlabelname{\textsc{#1}}%
}

\DeclareDocumentCommand \redruleAux {m m m}{%
	\textsc{#3}: \;\inferrule{}{#1 \redCCIC #2} \my@name@inferrule{#3}
}

\DeclareDocumentCommand \redrule {m m o}{%
	\IfValueTF{#3}
	{\redruleAux{#1}{#2}{#3}\hfill}
	{\inferrule{}{#1 \redCCIC #2}\hfill}
}

\DeclareDocumentCommand \Longredrule {m m o}{%
	\IfValueTF{#3}
	{\textsc{#3 :} \hfill \vspace{-0.5em} \\ \redrule{#1}{#2}}
	{\redrule{#1}{#2}}
}

\newcommand*{\ilabel}[1]{%
  \edef\@currentlabel{#1}
  \phantomsection
  \label{#1}
}

\makeatother

\newcommand{\prop}{\mathbb{P}}
\newcommand{\botProp}{\bot}
\newcommand{\topProp}{\top}
\newcommand{\BoxProp}{\ensuremath{\mathbb{B}\mathtt{ox}}}
\newcommand{\boxK}{\mathtt{box}}
\newcommand{\exfalso}[2]{\operatorname{\mathrm{exfalso}}_{#2} #1}
\DeclareDocumentCommand\sort{o}{\mathbb{s}\IfValueT{#1}{_{#1}}}

\newcommand{\precision}[2]{\tensor*[_{#1}]{\sqsubseteq}{_{#2}}}
\newcommand{\precisionDiag}[1]{\precision{#1}{#1}}
\newcommand{\precisionType}[1]{\sqsubseteq_{#1}}
\newcommand{\selfprecision}[2]{#1^{\precision{}{#2}}}
\newcommand{\selfprecisionType}[2]{#1^{\precisionType{#2}}}

\newcommand{\monarrow}{\rightarrowtriangle}
\newcommand{\notprecision}[2]{\tensor*[_{#1}]{\not\sqsubseteq}{_{#2}}}
\newcommand{\notprecisionDiag}[1]{\notprecision{#1}{#1}}

\newcommand{\congTm}[1]{{#1}_{\sqsubseteq}}
\newcommand{\unkUpTm}[1]{{#1}_{\?{}}}
\newcommand{\reflPrec}[1]{{\sqsubseteq}\texttt{-refl}_{#1}}
\newcommand\errMin[0]{\err{}\texttt{-min}}
\newcommand\unkMax[0]{\?\texttt{-max}}
\newcommand\propIrr[0]{\prop\texttt{-irr}}
\newcommand\boxPrecIrr[0]{\texttt{box}^{\sqsubseteq}_{\texttt{irr}}}

\newcommand{\bcatch}{\mathop{\catch_\bool^{[]}}}
\newcommand{\bcatchProp}{\mathop{\catch_\bool^\prop}}
\newcommand{\ncatch}{\mathop{\catch_\nat^{[]}}}

\newcommand{\lcatch}[1]{\mathop{\catch_{\listT~#1}^{[]}}}
\newcommand{\lind}[1]{\mathop{\ind_{\listT~#1}^{[]}}}
\newcommand{\lcatchProp}[1]{\mathop{\catch_{\listT~#1}^\prop}}

\newcommand{\boxcatch}[1]{\mathop{\catch_{\BoxProp~#1}}}

\newcommand{\lrefl}[1]{\left \lfloor{} {#1} \right \rfloor}
\newcommand{\urefl}[1]{\left \lceil{} {#1} \right \rceil}
\newcommand{\spType}[2]{\selfprecisionType{#1}{#2}}
\renewcommand{\sp}[2]{\selfprecision{#1}{#2}}
\newcommand{\prectrans}{\cdot}

\newcommand\upperdecompl[0]{\texttt{upper-decomp}}

\newcommand\ifte[3]{\text{if}~#1~\text{then}~#2~\text{else}~#3}

\newcommand{\transport}[4]{\mathrm{transport}~#1~#2~#3~#4}

\newcommand\len{\mathtt{len}}
\newcommand\sizedlist{\mathtt{Sized}\listT}
\newcommand\sizedlistprec{{\sizedlist_{\precision{}{}}}}

\newcommand\shiftTl[1]{[#1]}

\newcommand\wit[0]{\star}

\newcommand{\Ctx}{\mathcal{C}}

\usepackage{ifthen}

\newboolean{cameraReady}

\setboolean{cameraReady}{true} 

\ifthenelse{\boolean{cameraReady}}{
  \newcommand{\shepherd}[2]{#2}
  \newcommand{\shnew}[1]{#1}

  \newcommand{\mynote}[2]{}
}{
  \usepackage[normalem]{ulem}
  \newcommand{\shepherd}[2]{
  {\color{cyan!50!black}#2}}
  \newcommand{\shnew}[1]{{\color{cyan!50!black}#1}}

  \newcommand{\mynote}[2]
  {\ifdraft\textcolor{RedOrange}{\fbox{\bfseries\sffamily\scriptsize#1}
      {\small$\blacktriangleright$\textsf{\emph{#2}}$\blacktriangleleft$}}~\fi} 
}

\newcommand{\et}[1]{\mynote{ET}{#1}}
